%% file: 00.main.tex
\documentclass[12pt, journal, onecolumn, doublespacing]{IEEEtran}
\usepackage{booktabs} 
\usepackage{ifthen}
\usepackage{epsfig, times}
\usepackage{amssymb,amsmath,amsfonts}
\usepackage{setspace}
\usepackage{array}
\usepackage{algorithm}
\usepackage{balance}
\usepackage{graphicx}
\usepackage{changes}
\usepackage{caption}
\usepackage{makecell}
\usepackage{tabularx}
\usepackage{color}
\usepackage{adjustbox}
\usepackage{multirow}
\usepackage{balance}
\usepackage{soul}
\usepackage{algpseudocode}
\definechangesauthor[color=blue]{B}

\usepackage{amsthm}
 
\theoremstyle{definition}
\newtheorem*{definition}{Definition}

\usepackage[caption=false]{subfig}

\long\def\symbolfootnote[#1]#2{\begingroup%
\def\thefootnote{\fnsymbol{footnote}}\footnote[#1]{#2}\endgroup}

\newtheorem{theorem}{Theorem}
\newtheorem{lemma}{Lemma}

\newtheorem{remark}{Remark}[section]

\newcommand*\diff{\mathop{}\!\mathrm{d}}

\usepackage{changes}
\usepackage{picins}
\definechangesauthor[color=orange]{M}
\definechangesauthor[color=green]{G}

\newcommand{\fref}[1]{Fig.~\ref{#1}}
\newcommand{\sref}[1]{Section~\ref{#1}}
\newcommand{\cref}[1]{Chapter~\ref{#1}}

\newcommand{\eref}[1]{(\ref{#1})}

\newcommand{\lref}[1]{Lemma~\ref{#1}}
\newcommand{\thref}[1]{Theorem~\ref{#1}}

\newcommand{\aref}[1]{Appendix~\ref{#1}}

\usepackage{amssymb,amsmath, graphicx}
\usepackage{tikz}

\newcommand{\biagio}[1]{\textcolor{black}{#1}}
\newcommand{\gianluca}[1]{\biagio{#1}}

\def\bearn{\begin{eqnarray*}}
\def\eearn{\end{eqnarray*}}
\def\bear{\begin{eqnarray}}
\def\eear{\end{eqnarray}}
\def\barr{\begin{array}}
\def\earr{\end{array}}
\def\bmat{\left(\begin{array}}
\def\emat{\end{array}\right)}

\def\bd{\begin{definition}}
\def\ed{\end{definition}}
\def\bt{\begin{theorem}}
\def\et{\end{theorem}}
\def\be{\begin{center}\begin{equation}}
\def\ee{\end{equation}\end{center}}
\def\bc{\begin{corollary}}
\def\ec{\end{corollary}}
\def\bl{\begin{lemma}}
\def\el{\end{lemma}}
\def\br{\begin{remark}}
\def\er{\end{remark}}

\usepackage[noadjust]{cite}
\usepackage{graphicx}
\usepackage{caption}
\usepackage{subfig}
\usepackage{multirow}
\usepackage{ulem}

\usepackage[switch]{lineno}
\usepackage{enumitem}
\newtheorem{problem}{Problem}
\usepackage{nccmath}
\usepackage{epsfig}
\setlist[itemize]{leftmargin=*}
\newtheorem{corollary}{Corollary}[theorem]
\begin{document}

\title{Green Operations of SWIPT Networks: \\The Role of End-User Devices}
\author{\IEEEauthorblockN{Gianluca Rizzo,}
\IEEEauthorblockA{HES-SO Valais, Switzerland \& Universita' di Foggia, Italy - 
gianluca.rizzo@hevs.ch}\\
\and
\IEEEauthorblockN{Marco Ajmone Marsan, }
\IEEEauthorblockA{Institute IMDEA Networks, Spain - marco.ajmone@polito.it}\\
\and 
\IEEEauthorblockN{Christian Esposito, Biagio Boi,}
\IEEEauthorblockA{Universita' di Salerno, Italy - [esposito, bboi]@unisa.it}
}

\maketitle

\input{01.abstract}
\input{02.Introduction}

\input{02.Sota}
\input{03.Systemodel}

\input{04.Model}

\input{05.Optimization}
\input{05.1.optimization_problem_study}
\input{06.Numerical}
\input{08.Conclusions}
\bibliographystyle{IEEEtran}
\bibliography{IEEEabrv,refs,MyLibrary}
\input{Appendix}
\balance
\end{document}

%% file: 01.abstract.tex
\begin{abstract} 
Internet of Things (IoT) devices often come with batteries of limited capacity that are not easily replaceable or rechargeable, and that constrain significantly the sensing, computing, and communication tasks that devices can perform.
The Simultaneous Wireless Information and Power Transfer (SWIPT) paradigm addresses this issue by delivering power wirelessly to energy-harvesting IoT devices with the same signal used for information transfer. For their peculiarity, these networks require specific energy-efficient planning and management approaches. However, to date, it is not clear what are the most effective strategies for managing a SWIPT network for energy efficiency. In this paper, we address this issue by developing an analytical model based on stochastic geometry, accounting for the statistics of user-perceived performance and base station scheduling. We formulate an optimization problem for deriving the energy-optimal configuration as a function of the main system parameters, and we propose a genetic algorithm approach to solve it. Our results enable a first-order evaluation of the most effective strategies for energy-efficient provisioning of power and communications in a SWIPT network. We show that the service capacity brought about by users brings energy-efficient dynamic network provisioning strategies that radically differ from those of networks with no wireless power transfer.

\end{abstract}

%% file: 02.Introduction.tex
\section{Introduction}
\footnotetext{This paper was presented in part at the 20th International Symposium on Modeling and Optimization in Mobile, Ad Hoc, and Wireless Networks (WiOpt) \cite{9930622}.}
The advent of the Internet of Things (IoT) paradigm has significantly impacted the energy consumption of Radio Access Networks (RAN) \cite{Rajab2023}. On the one side, the proliferation of IoT devices is leading to an exponential increase in the number of connected devices and in the amount of data generated and exchanged \cite{PrepareY71:online}. On the other, the wide heterogeneity of IoT communication requirements and traffic patterns, together with the irregularity in the profiles of device activity and sleep cycles, induce patterns of energy consumption within the RAN that differ substantially from those of traditional, broadband-centric cellular networks \cite{TAHAEI2020102538,wan9206046}, with a strong impact on the way the network is planned and operated.\\
This impact is further amplified when the cellular network does not only deliver connectivity to all its users but also power to IoT devices. In this work, we consider scenarios in which IoT devices harvest RF energy from the environment (\textit{passive energy harvesting} - EH) as well as from the signal transmitted by their serving base station (\textit{active EH}) \cite{sanislav2021energy}. Such configurations are interesting as they allow for overcoming the power budget constraints that limit the potential of battery-operated IoT devices in many application domains. Indeed, these limitations often imply tight constraints on the amount of sensing, computing, and actuation that IoT devices can perform, negatively affecting their availability.

Among the available technologies for wireless power delivery to EH devices, simultaneous wireless information and power transfer (SWIPT) \cite{varshney2008transporting,SWIPT_costanzo2021evolution,SWIPT_huang2017simultaneous,SWIPT_ozyurt2022survey,SWIPT_perera2017simultaneous} is of particular interest, as it exploits the same cellular infrastructure used to deliver data. This feature makes it easier and more cost-effective to deploy than alternative wireless power delivery approaches. However, differently than  traditional single-service cellular networks with no active power transfer, managing SWIPT networks for energy efficiency implies accounting for the interdependence between the energy consumed by BSs and that consumed by IoT users which, for their ubiquity and rising numbers, are poised to play a key role in determining the overall energy footprint of the network. 
Thus, SWIPT networks require the development of specific approaches for energy-efficient network design and management, which may substantially differ from traditional, single-service RANs.

The issue of energy efficiency in SWIPT is being actively investigated \cite{zhou_wireless_2013,huang_energy-efficient_2018,Tran8388301}. However, these works consider simple network scenarios, composed of a single base station (BS). They focus on accounting for the impact of the interplay between power and information delivery on the energy footprint of the BS. Other works (\cite{lu_energy_2021,Lee9031335}) consider a set of BSs deployed according to a specific layout, and they propose heuristics for saving energy in those configurations. Being tied to the specifics of the scenario considered, these heuristic approaches do not allow drawing general conclusions on the main features and limitations of energy-optimal management strategies in SWIPT networks.
Another line of research focuses on applying the tools of stochastic geometry (SG) \cite{lam2016system,di_renzo_system-level_2017,Baccelli_stochasticgeometry} for modeling system-level average performance in SWIPT networks. \biagio{SG allows focusing on the average behavior of the system over many realizations of the process of UE (user equipment) and BS spatial distributions, making it possible to analyze the trade-offs between different potential objectives such as maximizing coverage, maximizing throughput, and maximizing energy harvesting, in ways that would otherwise be very difficult or way less accurate.} Crucially, however, none of these results account for the effects of resource scheduling among users on user-perceived performance, nor do they address the key issue of characterizing  energy-optimal strategies for QoS-aware (for both power and information delivery) dynamic network provisioning, as a function of traffic and energy demand. \biagio{Thus, it is still unclear what is the actual potential for energy efficiency of strategies which dynamically tune the configuration of SWIPT networks, and what is the impact of optimizing the main system parameters} on the energy consumed by IoT users and BSs.

In this paper, we address these issues and propose a stochastic geometry analytical framework for modeling the relationship between energy consumption, user perceived performance, and the main system parameters in SWIPT cellular networks, which accounts for QoS-aware resource scheduling among broadband and IoT EH users. Our framework allows characterizing in a scenario-independent manner the potential for QoS-aware energy efficiency in these networks, as well as the main trade-offs between user-perceived performance in the communication and power delivery services, resource utilization, and overall energy consumption of the network. Specifically, the contributions of this paper are as follows:
\begin{itemize}
    \item We derive a set of analytical results for the key statistics of the main performance indicators of a SWIPT network, which account for resource scheduling among broadband and IoT EH users, and allow modelling the impact of the main system parameters on network energy consumption;
    \biagio{\item We formulate an optimization problem to determine the potential energy savings achievable by tuning the main network parameters while guaranteeing a target user-perceived performance for both information transfer (in downlink and uplink) and wireless power delivery.}
    \item We elaborate a Genetic Algorithm (GA) approach to derive energy-efficient SWIPT network configurations which achieve the target QoS levels;
    \item We validate numerically our approach,  shedding light on several aspects of the tradeoff between user-perceived performance and resource efficiency in a SWIPT network. Our analysis allows observing some unexpected and rather surprising effects such as, in some cases, a decrease in the optimal density of BSs required to serve an increasingly dense population of UEs, and the irrelevance of active power delivery for very dense populations of IoT devices.
\end{itemize}
The rest of this paper is organized as follows. \sref{sota} reviews the relevant state of the art. In \sref{model}, we present the system model, and in \sref{sec:model} our analytical approach for user-perceived performance modelling for both information and power transfer. In \sref{sec:optimizing}, we present our formulation of the optimization problem and the GA heuristic for solving it efficiently. In \sref{sec:numerical} we assess numerically our results. Finally, \sref{sec:concl} concludes the paper.
%



%% file: 02.Sota.tex
\section{Related works}
\label{sota}


The growing interest towards the integration of wireless power transfer in cellular networks is due to its potential for enabling truly ubiquitous and self-sustained cellular-based IoT \cite{varshney2008transporting,SWIPT_costanzo2021evolution,SWIPT_huang2017simultaneous,SWIPT_ozyurt2022survey,SWIPT_perera2017simultaneous}. Indeed, differently from other energy harvesting techniques (e.g., based on solar, or solely on passive RF harvesting) those based on active wireless power transfer have the advantage of being stable and available at any time \cite{sanislav2021energy}.
Several technical challenges however still need to be addressed for the practical viability of SWIPT networks, related to energy-harvesting transceiver and algorithm design, system integration, protocol design, and energy-efficient network planning and operations \cite{niyato2017wireless,Clerckx8476597}.

Initial works on energy efficiency in SWIPT networks considered scenarios composed
of a BS or a broadband unit with
several remote radio heads, and a population of EH devices. These works aim to optimize the system configuration in terms of energy efficiency  \cite{zhou_wireless_2013,huang_energy-efficient_2018, luo_study_2021}. Among these, \cite{zhou_wireless_2013} characterizes the data rate vs. energy trade-off in a single SWIPT BS scenario. For the same setting, \cite{huang_energy-efficient_2018} proposes an approach for optimizing transmit power allocation to each user.  
\cite{luo_study_2021} elaborates a Markov chain model to derive the joint power and connectivity outage probability. \cite{havutran} 
investigates strategies for transmit-power-efficient resource allocation. 
All of these works provide a first insight into the basic performance patterns of a SWIPT system. However, they are based on single BS configurations, which are not representative of the average performance of a whole SWIPT network, and they do not account for the impact of the statistics of BS and user distribution. In particular, they lack accurate modelling of the effects of co-channel interference and passive RF power transfer from BSs other than the serving one, as well as from users associated with them.

Another set of results focuses on sample SWIPT network configurations, with a given set of BS and a given spatial layout of BS and of users. They propose algorithms for energy-optimal configurations, for networks with CoMP \cite{tang_energy_2018-1}, in NOMA networks \cite{tang_energy_2019}, in full duplex settings \cite{yuan_energy_2019}, as well as in scenarios with two-way relays \cite{li_robust_2020}. Among these, \cite{lu_energy_2021} proposes an algorithm to minimize the total energy consumption of a set of BSs and a population of IoT devices, over sub-carrier and transmit power allocation.
However, the strategies proposed in these works focus on deriving the optimal configuration for a sample BS deployment, with a specific layout and with a given user spatial distribution. Being tied to a specific setup, they do not allow drawing general considerations on the overall potential of these schemes for energy-saving. 

Recently, stochastic geometry (SG) 
\cite{lam2016system,di_renzo_system-level_2017,Baccelli_stochasticgeometry,Chu9463400} 
 has emerged as an effective modelling approach for a stochastic characterization of the performance patterns in a wireless network. 
  SG allows focusing on the average behaviour of the system over many realizations of the process of UE (user equipment) and BS spatial distributions. In SWIPT networks, SG has been used to characterize several performance aspects, such as the data rate vs. energy trade-off \cite{di_renzo_system-level_2017}, or the optimization of the D2D successful transmission rate \cite{Chu9463400}, to name a few.
These works however, often for the sake of analytical tractability, do not account for the effects of resource allocation and scheduling among users (such as the statistics of the sharing of BS time across all associated users) which are key for accurate stochastic modelling of the performance perceived by the user. As such, they do not enable an accurate and realistic characterization of the main trade-offs between the energy consumed by the network, user-perceived performance (for both power and information delivery), and resource utilization.
\biagio{Thus, they leave open the more general issue of determining the most effective QoS-aware dimensioning and tuning strategies to achieve system-level energy proportionality in a SWIPT network, as a function of the main system parameters.}

\biagio{In the present paper, we address these issues by proposing an SG modelling framework which captures the relationship between energy consumption, user-perceived performance, and the main system parameters in SWIPT cellular networks. 
Our framework allows characterizing in a scenario-independent manner the potential for QoS-aware energy efficiency in these networks, as well as the main trade-offs between user-perceived performance in the communication and power delivery services, resource utilization, and overall energy consumption of the network. Our analytical approach overcomes the limitations of existing results, as follows:
\begin{itemize}
    \item It allows an estimation of the potential energy savings achievable in a SWIPT network through QoS-aware (for both energy and data transfer) tuning of network parameters.
    \item It accounts for the effects of base station resource allocation and scheduling among users on the statistics of user-perceived QoS.
\end{itemize}
In doing so, our work allows a first characterization of the performance-energy tradeoff, revealing some new and rather unexpected patterns related to the impact of the population of user devices on the overall energy efficiency of the network.} 
 

%% file: 03.Systemodel.tex
\section{System Model}
\label{model}


\biagio{We utilize stochastic geometry to model the spatial distribution of network elements in our system}. Specifically, we assume that BSs are distributed in space according to a homogeneous planar Poisson Point Process (PPP) with a density equal to $\lambda_b$ BSs per $km^2$. UEs are distributed in space according to a homogeneous PPP with an intensity equal to $\lambda_u$ UEs per $km^2$. UEs are either broadband (BB) terminals, or  IoT (Internet of Things) devices. We assume the latter are a fraction $\gamma$ of the total number of UEs.

\begin{figure}[!t]
\centering
\includegraphics*[width = 4.5in]{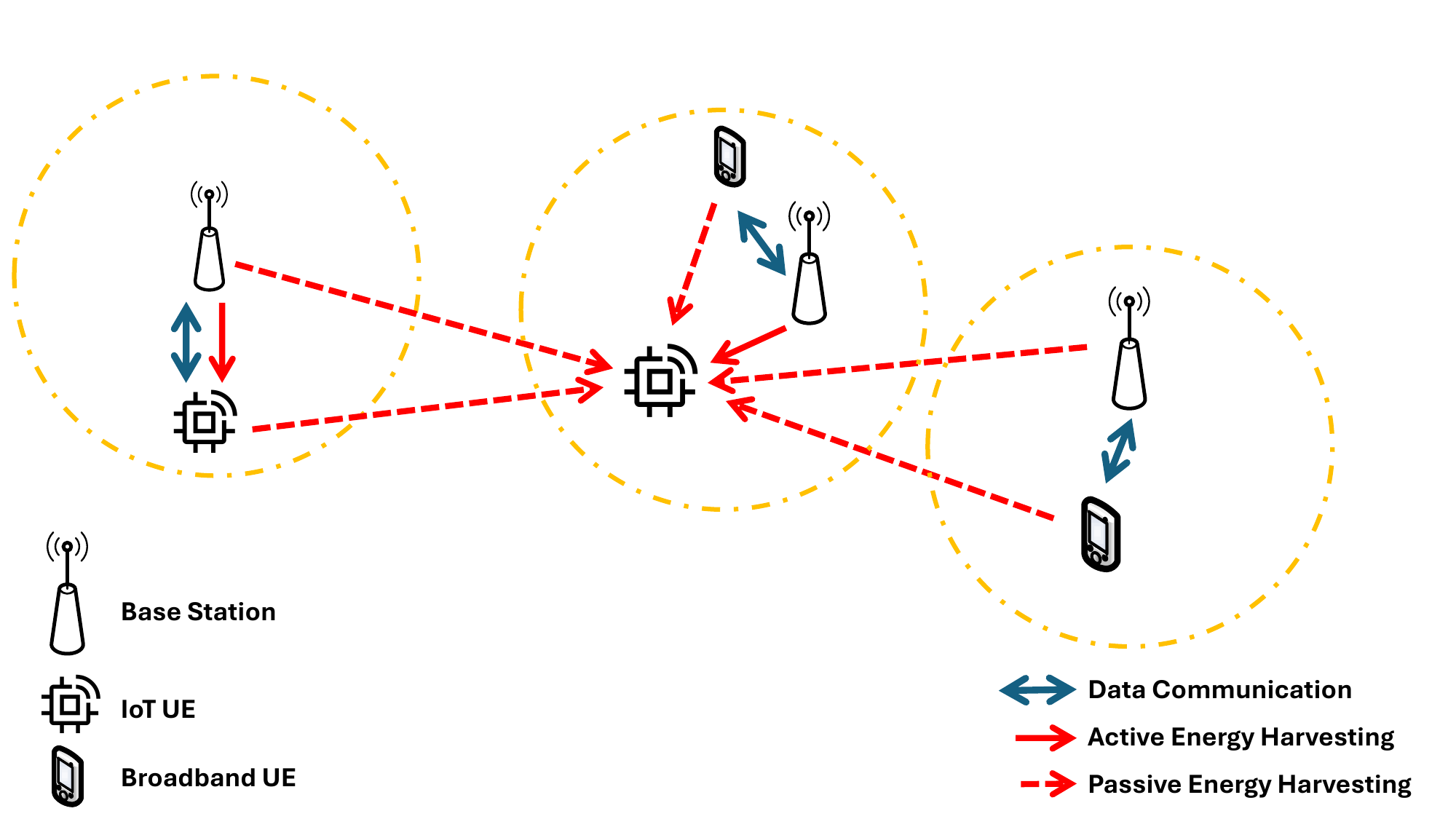}
\caption{\footnotesize Outline of the system model for the SWIPT wireless network considered in this work.\normalsize} 
\label{fig:system_outline}
\end{figure}

We consider the case in which IoT devices harvest the energy necessary for their own operation, while BB users don't. To this end, each IoT device is equipped with two separate receivers, 
one for information and another for energy, and it is capable of exploiting downlink signals for decoding its intended information, as well as downlink and uplink signals from both BSs and UEs for charging its battery. 
The receiver operating modes we consider are~\cite{nasir2013relaying}:
\begin{itemize}
    \item \textbf{Time Switching (TS)}, by which a fraction $\mathbf{\eta}$ (the \textit{time switch ratio}, with $0\leq\eta\leq 1$) of the time dedicated by a BS to serve an associated IoT device in downlink is devoted to active power transfer, i.e., it is used by the UE for harvesting energy from the signal received from the BS. The rest of that time is used for receiving information.
    \item\textbf{Static Power Splitting (SPS)}. In this operating mode, the receiver antenna at every IoT device is followed by a splitter, which sends a fraction $\mathbf{\nu}$ (the \textit{power split ratio}, with $0\leq\nu\leq 1$) of the total signal power received at any time instant to the RF harvesting electronics. The rest of the received signal power is instead used for decoding information.
    \item \textbf{Dynamic Power Splitting (DPS)}. In this operating mode, the receiver architecture is the same as in the SPS mode. The only difference is that, in DPS, at every IoT user the power split ratio is $\mathbf{\nu}$ when the serving base station is transmitting to that user, and it is equal to one for the rest of the time. That is, when an IoT node is not receiving data from its serving base station, all the received power (from any RF source, including all BSs and users) is fed to the RF harvesting electronics.
\end{itemize}
We assume $\eta$ and $\nu$ are the same for all devices.
Thus, IoT devices harvest energy not only from their serving BS (\textit{active charging}), but also from the signal received from all of the other BSs, as well as from uplink transmissions from both IoT and BB UEs (\textit{passive charging}).
\begin{figure}[!t]
\centering
\includegraphics*[width = 4.5in]{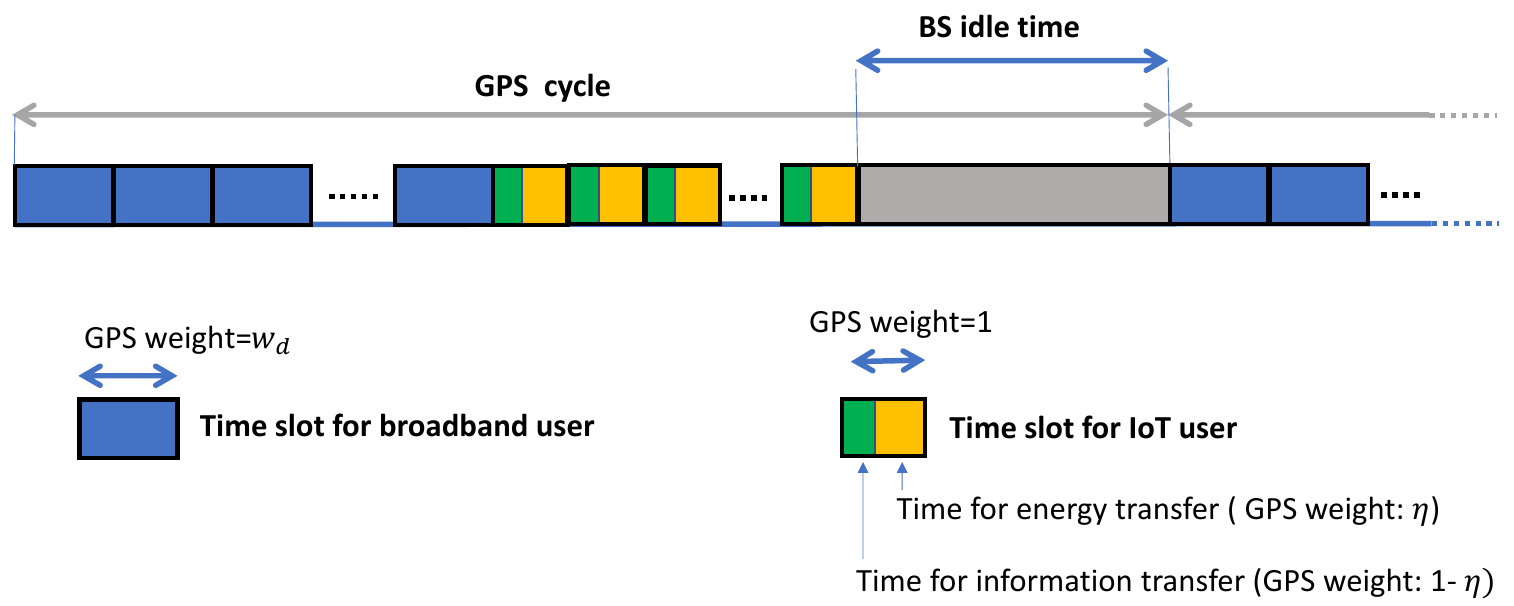}
\caption{\footnotesize Downlink time scheduling scheme for a SWIPT BS, for a network with time switching EH receiver architecture, with GPS weights.\normalsize} 
\label{fig:scheduling}
\end{figure}

We assume that BSs use a generalized processor sharing (GPS) mechanism to divide BS time among all the connected UEs. 
In downlink, 
the GPS weights are $1$ for IoT UEs, and 
$w_d$ for BB UEs. 
The time spent by the BS without transmitting is modeled as a user with a GPS weight $\beta_{d}$.
As for uplink, in all configurations, the GPS weights are $1$ for IoT UEs, $\delta_u$ for BB UEs, and $\beta_{u}$ for the uplink BS time not assigned to any UE.
IoT UEs periodically cycle between two operational states. During the \textit{active state}, they send and receive data, plus possibly they perform some other task, such as sensing, while harvesting energy from active and passive sources. In the \textit{low power} state, IoT devices only harvest energy. We assume that each IoT device is disconnected from the power grid, and possesses an ideal battery (i.e., with efficiency equal to one) with enough capacity to compensate for the fluctuations in the energy consumed and harvested. Note however that our approach can be easily extended to consider battery nonidealities and limitations due to finite battery capacity. With $\phi$ we denote the fraction of time spent by an IoT device in the active state. BB users are instead assumed to always be in the active state.\\
We assume the system is in saturation, i.e., BSs have always data to send to users, and users in the active state have always data to transmit. 
A key aspect of the EH process is the function $h=\Theta(h_{in})$ that maps the received power $h_{in}$ to the harvested power $h$ for EH IoT devices. $\Theta$ models the effect of several technological factors affecting the efficiency of the energy harvesting process, ranging from receiver antenna gain to exposure of the specific IoT device, to receiver architecture, among others. In this work, we consider two different models:
\begin{itemize}
    \item a linear model, by which $\Theta(h_{in})=\xi h_{in}$, with $\xi\leq 1$; and
    \item a nonlinear model, where the mapping function is a normalized sigmoid \cite{wang_wirelessly-powertwoway_2017}, given by   \begin{equation}\label{eq:sigmoid}
        \Theta(h_{in}) = \max\left\{\frac{h_{max}}{e^{-\chi h_{s}+\iota}} \left(\frac{1+e^{-\chi h_{s}+\iota}}{1+e^{-\chi h_{in}+\iota}}-1\right), 0\right\}
    \end{equation}
     $h_{max}$ is the maximum harvested power when the energy harvesting circuit is saturated, while $h_{s}$ is the harvester's sensitivity threshold, i.e. the minimum amount of power able to activate the harvesting process. $\chi$ and $\iota$ are parameters that control the steepness of the sigmoid.
\end{itemize} 
Note however that our approach extends easily to other EH models.\\
\begin{table}
    \centering
    \caption{Main notation used in the paper}
    \footnotesize
    \label{tab:symbols}
    \begin{tabular}{lp{15cm}}
        \hline
        \textbf{Name} &   \textbf{Description} \\
        \hline
$\bar{\tau}_j$, $j=d,u$ & Mean ideal per-bit delay  in downlink/uplink ($s^{-1}$)\\
$\tau_j$, $j=d,u$  & Per-bit delay  in downlink/uplink ($s^{-1}$)\\  
$\tau_j^0$, $j=d,u$ & Target per-bit delay in downlink/uplink ($s^{-1}$)\\
$\xi$ & EH conversion efficiency ratio (linear model)\\
$h_0$ & Minimum harvested power required by each IoT device ($W$)\\
$\mu$ & Maximum acceptable ratio of IoT users which harvest less than $h_0$ W\\
$\phi$ & Fraction of time in which IoT devices are active\\
$\lambda_b$ & Mean BS density ($m^{-2}$)\\
$\lambda_u$ &  Mean density of user terminals ($m^{-2}$)\\
$\gamma$ & Fraction of user terminals which are IoT devices\\
$\eta$ & Fraction of the time dedicated by a BS to serve an IoT device in downlink used for harvesting energy from the signal received from the base station\\
$\nu$ & Fraction of the power received at any time instant by an IoT device which is used for energy harvesting\\
$\beta_j$, $j=d,u$ & GPS weight of the amount of time spent by a  BS in idle mode in downlink/uplink\\
$w_d$ & Ratio between the amount of time dedicated to serving a BB UE and the one dedicated to serving an IoT user in downlink\\
$\delta_j$, $j=d,u$  & Ratio between the BS time dedicated to a BB user, and the base station time dedicated to an IoT user for data transfer in downlink/uplink\\
$k$ & Frequency reuse factor\\
$\alpha$ & Path loss exponent\\  
$P$ & BS transmit power ($W$)\\
$G$ & BS antenna gain\\
$L$ & BS antenna loss (out of beam)\\
      \hline
    \end{tabular}
    \normalsize
\end{table}
Given that the cellular RAN delivers two classes of services to two different types of UE, the performance metrics are defined as follows:
\begin{itemize}
    \item \textbf{Information transfer}: For all UEs, the end-user performance metric is the \textit{per-bit delay $\tau$} of data transfers. It is the time required to transfer a single bit, and it is thus the inverse of the short-term throughput, i.e., of the rate at which data is transferred.
    \item\textbf{Power transfer}: The performance parameter is the amount of power harvested by an IoT device, denoted with $h$.
\end{itemize}
As a result of the scheduling strategy (both in downlink and uplink), the amount of harvested energy may vary substantially across the various devices.
These metrics must be considered relative to the corresponding application demands. As an example, the energy harvested by each device during a complete on-off cycle (and therefore the value of the parameter $h$) must be sufficient to compensate for the energy consumed during the cycle by such tasks as sensing, processing, storing, and distributing data, while accounting for battery efficiency and for fluctuations in consumption patterns due to changes in the operational state of IoT devices. Accordingly, the target performance for information transfer is related to application-specific requirements.

\subsection{Channel and Service Model}
Our channel model only takes into account distance-dependent path loss. {Incorporating the effects of fading and shadowing would not alter our approach, and is left for future consideration.} We assume that \textit{random frequency reuse} is in place, with reuse factor $k$. That is, every BS is assigned one out of $k$ frequency bands with equal probability.\\
We assume that UEs are associated with the BS that provides the largest SINR at the user location. We consider urban scenarios, where the high capacity demand justifies strategies for energy-efficient network planning and management, and where the assumption of high attenuation (with exponent $\alpha\geq 3$) typically holds. In these settings, as no fading is considered and all BSs use the same transmit power, assuming that users associate to the closest BS is a reasonable approximation \cite{Baccelli_stochasticgeometry}. 
\\ 
We assume BS antennas use beamforming, and we denote with $G$ the beamforming gain and with $L$ the side lobes attenuation.
If with $a$ we denote the 
aperture of the main lobe (in degrees),
the relationship between $G$ and $L$ is given by
\[
L= 1-(G-1)\frac{a}{360-a}
\]
Since users are assumed to be distributed uniformly in space, and base station time is supposed to be shared in equal parts across users of the same type, 
the mean power received by a user while it is not being served from the base station to which it is associated (located at a distance $r$) is $PL_g r^{-\alpha}$, with $L_g=\frac{Ga+L(360-a)}{360}$.
Thus $L_gr^{-\alpha}$ is the mean attenuation with which the power transmitted by the serving base station is received by a user when it is not being served.

Denote by $S(x)$ the location of the BS that is closest to a UE located at $x$.
We denote the capacity of a user located at a distance $r$ from the BS by 
$C(r,P,G,I)$ bit/s per Hertz, where $P$ is the BS transmit power, and $I$ the total received interfering power. 
We model $C(r,P,G,I)$ using Shannon's capacity law: 
\vspace{-0.3in}
\be\label{eq:shannon}
C(r,P,G,I) = \frac{B}{k} \log_{2}\left(1 + \frac{P G r^{-\alpha}}{N_0+I(r,k)}\right)
\ee
where $\alpha$ is the attenuation coefficient, $N_0$ the power spectral density of the additive white Gaussian noise, and $k$ the reuse factor. 

\subsection{Base Station and UE Energy Consumption Model}
\label{sec:energymodel}

The power consumed by a BS depends on several factors, which vary according to the BS type (e.g. macro, micro, femto) and the implementation technology (e.g. standalone vs cloud-RAN), among others. 
In what follows, we adopt a very flexible BS energy model, first proposed in \cite{futureproof2015,mushtaq_power_2017}, by which the power consumed by a BS, denoted as $P_{BS}$, is given by the following expression:
\vspace{-0.3in}
\be\label{eq:energymodel}
P_{BS}=q_1 + U_d[q_2+q_3 (P-P_{min})]
\ee
In this model, the power consumed is thus given by the sum of three contributions.
\begin{itemize}
    \item A first contribution, given by $q_1$, is constant, and depends only on the specific type of BS considered (e.g. macro, micro, femto, cloud). It models the power consumed when the BS is not carrying any traffic to/from users. As such, it does not depend on utilization or transmit power, and it is due to the power consumed by, e.g., part of the cooling function, by power amplifier consumption in idle state, and by all those functions which keep the BS in an operational (i.e. non standby) state. 
    \item The second component $q_2 U_d$, in which   $U_d$ is the downlink BS utilization, models that fraction of consumed power which depends on BS utilization (and it is thus proportional to the amount of traffic served), but not on transmit power. It models the power consumed by such functions as baseband processing and RF signal processing.
\item Finally, the third component $ q_3 U_d (P-P_{min})$ (where and $P$ is its transmit power, which we assume varies within the interval $[P_{min},P_{max}]$) models the fraction of consumed power due to the power amplifier which depends, at the same time, on the transmit power $P$ and on the fraction of time that the BS is busy transmitting (given by $U_d$).
\end{itemize}
This energy model is very flexible and suited for accounting not only for implementations with different degrees of load proportionality, but also for cloud RAN configurations.

%% file: 04.Model.tex
\subsection{Base station service model}
\label{sec:servicemodel}
We define the utilization $U(S(x))$ of the BS serving a UE at location $x$ as the average fraction of time in which the BS is busy transmitting (in downlink -- $U(S_d(x))$) or receiving (in uplink -- $U(S_u(x))$), respectively. Thus, the expression of the utilization of the BS is given by the fraction of BS time dedicated to all active users:
\vspace{-0.3in}
\be\label{eq_utilization_1}
U_d(S(x))=\frac{N_{iot}(S(x))+w_d N_{bb}(S(x))}{N_{iot}(S(x))+w_d N_{bb}(S(x))+\beta_{d}}
\ee 
\vspace{-0.15in}
\be\label{eq_utilization_2}
U_u(S(x))=\frac{N_{iot}(S(x))+\delta_u N_{bb}(S(x))}{ N_{iot}(S(x))+\delta_u N_{bb}(S(x))+\beta_{u}}
\ee 
where $N_{iot}(S(x))$ and $N_{bb}(S(x))$ denote respectively the number of broadband and IoT users in the \textit{active} state, and associated with the same BS as the UE at location $x$. Note that, as IoT users cycle between active and inactive states, on average only a fraction $\phi$ of them is active at any point in time.
Given $w_d$ and $\delta_u$, by tuning $\beta_j$ it is possible to vary the mean amount of service received by UEs for both communication and energy transfer, and the overall BS utilization, both in downlink and in uplink. 

\begin{definition} The ideal per-bit delay perceived by a UE is the per bit delay which a UE would perceive if the BS with which the UE is associated had utilization equal to 1.
\end{definition}
Note that the above definition does not assume that \textit{all} of the BSs have utilization equal to one, but only the BS serving the considered UE. From expressions (\ref{eq_utilization_1}), (\ref{eq_utilization_2}) and the above definition, we have that for a UE at $x$, the relationship between the ideal per bit delay $\tau_j^{id}$, $j\in\;d,u$ and the actual per bit delay, both in uplink and in downlink, is
$U_j(S(x))\tau_{j}(S(x))=\tau_j^{id}(S(x))$, with $j\in\;d,u$. Indeed, from the expression of the utilization, the ratio between the ideal and the actual per-bit delay is equal to the fraction of time the BS is active.

For each UE type, and for both uplink and downlink, a notion of target minimum quality of service (QoS) is defined. Namely, UEs are said to perceive satisfactory performance if the average per-bit delay experienced by a \textit{typical} BB user (resp. IoT device)\footnote{The definition of the typical user in the system is provided by classical Palm theory \cite{Stoyan1987}.} are less than their respective predefined target values. 
In what follows, for a given $\eta$, we assume each BS adopts the values of $\beta_j$, $j\in\;d,u$ (and thus of utilization in downlink and uplink) which achieve the target QoS for communications. 
This makes the average of the actual per bit delays (over all the users associated to that BS) coincide with the target values.


\section{Modeling User-Level Performance}
\label{sec:model}
In this section, we characterize the main performance parameters, i.e., the per-bit  delay and the harvested power of a typical user who is just beginning service\footnote{This assumption is standard in stochastic geometry, and it is one of the possible ways of choosing the point of view from which to carry out the analysis. Indeed, there is no difference in terms of performance with respect of the case in which users have been in the system for a while.}
, as well as the mean harvested power, as a function of the main system parameters.
The expression of the power harvested by a user at $x$ is given by the following result.
\begin{lemma}\label{lemma:perjouledelay} Let 
$K(x)=[N_{iot}(S(x))+w_d N_{bb}(S(x))]^{-1}$.
Then the power harvested by a user at $x$ is $h(x)=\Theta(h_{in}(x))$, where the received power $h_{in}(x)$ is given by: 
\begin{itemize}
    \item \textbf{TS}:
    \vspace{-0.1in}\begin{equation}\label{eq:harvested_power_TS}
h_{in}(x)=P D(x)^{-\alpha}U_d(S(x))\left[G\eta K(x)+L_g(1-K(x))\right]+\left[1-K(x)U_d(S(x))(1-\eta)\right](I(x) +O(x))
\end{equation}
    \item \textbf{SPS}:
\vspace{-0.1in}
\begin{equation}\label{eq:harvested_power_SPS}
h_{in}(x)=\nu P D(x)^{-\alpha}U_d(S(x))\left[G K(x)+L_g(1-K(x))\right]+\nu(I(x) +O(x))
\end{equation}
    \item \textbf{DPS}:
    \vspace{-0.1in}
\begin{equation}\label{eq:harvested_power_SPS}
h_{in}(x)= P D(x)^{-\alpha}U_d(S(x))\left[G\nu K(x)+L_g(1-K(x))\right] +I(x) +O(x)
\end{equation}
\end{itemize}

$I(x)$ is the total power harvested  by the user at $x$ from BSs other than the one with which it is associated. $O(x)$ is the power harvested from UE transmissions, averaged over time, and $U_d(S(x))$ is the downlink utilization of the BS with which the user at $x$ is associated.
\end{lemma}
For a proof, please refer to Appendix \ref{app:proof_lemma_perjouledelay}.
From its definition and from the service model description, it derives that in downlink, the ideal per-bit delay perceived by a BB UE at $x$ is given by 
\vspace{-0.1in}
\be\label{eq:taubroadband}
\tau_d^{id}(x)=\frac{N_{iot}(S(x))+w_d N_{bb}(S(x))}{w_d C(x,P,G,I)}
\ee
For IoT users, the per-bit delay is given by:
\begin{itemize}
    \item \textbf{TS}:
    \vspace{-0.3in}\be\label{eq:tau_iot_TS}
\tau_{d,I,TS}^{id}(x)=\frac{N_{iot}(S(x))+w_d N_{bb}(S(x))}{(1-\eta) C(x,P,G,I)}
\ee
    \item \textbf{SPS}, \textbf{DPS}:
\vspace{-0.3in}\be\label{eq:tau_iot_SPS}
\tau_{d,I,PS}^{id}(x)=\frac{N_{iot}(S(x))+w_d N_{bb}(S(x))}{ C(x,(1-\nu) P,G,(1-\nu)I)}
\ee
\end{itemize}
For the uplink instead, for broadband users, we have
\vspace{-0.15in}
\be\label{eq:taubroadband_uplink}
\tau_u^{id}(x)=\frac{ N_{iot}(S(x))+\delta_u N_{bb}(S(x))}{\delta_u C(x,P_{I},1,0)}
\ee
For IoT users, the ideal per-bit delay perceived in uplink is $\tau_{u,I}^{id}(x)=\delta_u\tau_u^{id}(x)$.
With $\bar{\tau}_j$ ($\bar{\tau}_{j,I}$), $j=d,u$ we denote the average per-bit delay perceived in downlink (resp. in uplink) by broadband and IoT users, respectively.
\biagio{By leveraging stochastic geometry, we now derive the main analytical results, in terms of probabilistic expressions for user-perceived performance metrics, which account for the random spatial patterns of both BSs and UEs.} 
The following result derives an expression for the average per-bit delay perceived by the typical user (BB or IoT) which is just beginning service.
\bt\label{th:BE_tau} 
The mean ideal per-bit delays in downlink and uplink, and the mean ideal per-Joule delay perceived by a typical best-effort user joining the system are given by:
	\begin{align}
\bar{\tau}_d =&   H(w_d,w_d,C(r,P,G, \bar{I}) )\label{eq:tau_1} \\
\bar{\tau}_{d,I,TS} =& \bar{\tau}_d \frac{w_d}{(1-\eta)} \label{eq:tau_2}\\
\bar{\tau}_{d,I,PS} =&   H(w_d,1,C(r,(1-\nu) P,G, (1-\nu)\bar{I}) ) \label{eq:tau_3}\\
\bar{\tau}_u =&   H\left(\delta_u,\delta_u,C(r,P_{I},1,0) \right) \\
\bar{\tau}_{u,I} =& \delta_u\bar{\tau}_u
	\end{align}
Where:
\begin{equation}\label{eq:Poisson}
H(y,z, g(r))=\int_0^{\infty} \frac{f(r,y)e^{- \lambda_b \pi r^2} \lambda_b 2 \pi r}{zg(r)}   \diff r.
\end{equation}
with
\vspace{-0.1in}
\[
f(r,y)=\lambda_{u}\left[y+\gamma\left(\phi-y\right)\right]\int_0^{\infty} \int_0^{2 \pi}e^{- \lambda_b A(r,x,\theta)} x d \theta d x 
\]
$A(r, x, \theta)$ is given by
$A(r, x, \theta)= \pi x^2 - \bigg[r^2 \arccos\left(\frac{r +  x \sin(\theta)}{d(r, x, \theta)}\right)+x^2\arccos \left(\frac{x + r \sin(\theta)}{d(r, x, \theta)}\right)+\bigg.$

$\bigg.-\frac{1}{2} \sqrt{[r^2-(d(r, x, \theta)-x)^2 ][(d(r, x, \theta) +x)^2-r^2] }\bigg]$, and 
$d(r,x,\theta)$ is the euclidean distance between $(x,\theta)$ and $(0,-r)$. $C(r,P,G, \bar{I})$ is given by \eref{eq:shannon}, with the interference term $\bar{I}$ given by
 \vspace{-0.1in}
 \be\label{eq:interference}
\bar{I}(r,k)=\frac{PL_g\lambda_b2\pi r^{2-\alpha}}{k(\alpha-2)}\frac{\bar{\tau}_d}{\tau_d^0}
\ee


\et
For the proof, please refer to appendix \ref{app:proof_th_BE_tau}.
We say that the considered system is in the \textit{dense IoT regime} when the density of IoT users is such that the probability of having a cell without IoT users is negligible.
{This assumption is coherent with many current projections on expected IoT deployments in 6G, in which the expected device densities are of the order of 10 million devices per km$^2$ and more~\cite{BANAFAA2023245,QADIR2023296,9369324,9509294}.}
\bt\label{th:BE_tau_stdev} In the dense IoT regime, the cumulative distribution function $CDF_h$ of the power harvested by an IoT user who is just beginning service is
 $ CDF_h(h_0)=CDF_r(g^{-1}(h_0))$ for all $h_0\geq 0$, where $CDF_r$ is the cumulative distribution function of the distance of the user to its serving BS:
  \[
 CDF_r(r)=\int_0^{r}e^{- \lambda_b \pi y^2} \lambda_b 2 \pi ydy
 \]
 where $g(r) = \Theta \left (F_j(r)+\frac{Z_j(r)}{f(r,w_d)}\right)$ 
with $j\in \{TS, SPS, DPS\}$, and $f(r,w_d)$ given by \thref{th:BE_tau}.
\begin{itemize}
    \item \textit{TS}:
    \[
F_{TS}(r)=P r^{-\alpha}L_g\frac{\bar{\tau}_d}{\tau_d^0} +k\bar{I}(r,k) +\bar{O}
\]
    \[
Z_{TS}(r)=\frac{\bar{\tau}_d}{\tau_d^0}\bigg[P r^{-\alpha}(G\eta-L_g)
-(1-\eta)(k\bar{I}(r,k) +\bar{O})\bigg]
\]
\[
\bar{O}=\frac{(1-\gamma)\delta_u P_{bb} +\phi \gamma P_I}{(1-\gamma)\delta_u+\phi\gamma}\frac{\lambda_b\pi \alpha }{\alpha-2}\frac{\bar{\tau}_u}{\tau_u^0}
\]
\item \textit{SPS}: $F_{SPS}(r)=\nu F_{TS}(r)$, and $Z_{SPS}(r)=\nu P r^{-\alpha}\frac{\bar{\tau}_d}{\tau_d^0}(G-L_g)$.
\item \textit{DPS}: $F_{DPS}(r)= F_{TS}(r)$, and $Z_{SPS}(r)=P r^{-\alpha}\frac{\bar{\tau}_d}{\tau_d^0}(\nu G-L_g)$.
\end{itemize}

\et
For the proof, please refer to \aref{app:proof_stdev}. $P_{bb}$ and $P_{I}$ denote the transmit power of BB and IoT UEs, respectively.
\begin{theorem}{\textbf{(Maximally fair GPS weights).}} The values of the GPS weight $w_d$ which maximize fairness among users are given by $w_d=\delta_d (1-\eta)$ for the time split (TS) mode, and 
\[
w_d=\delta_d \frac{\log_{2}\left(1 + \frac{(1-\nu) P G \bar{r}^{-\alpha}}{N_0+(1-\nu)\bar{I}(\bar{r},k)}\right)}{\log_{2}\left(1 + \frac{P G \bar{r}^{-\alpha}}{N_0+\bar{I}(\bar{r},k)}\right)}
\]
for the power split modes. $\bar{r}$ is the mean distance of UEs from the serving base station, and $\bar{I}(\bar{r},k)$ is given by \thref{th:BE_tau}. 
\end{theorem}
\begin{proof}
In our GPS scheduler, $w_d$ is the ratio between the amount of time dedicated to serving a BB UE and the one dedicated to serving an IoT user (for both information transfer and power transfer).
We need to ensure fairness in resource scheduling for information transfer between BB and IoT users (as for power transfer, only IoT users are involved). This is implemented by having the ratio between the (mean) amount of bits received by a BB user and an IoT user be equal to $\delta_d$, i.e. to the ratio between the target values of throughput of the two user classes.  
Thus we have, for the TS case, in any location at a distance $r$ from the serving base station,
\[
\delta_d=\frac{w_d R(r)}{(1-\eta)R(r)}
\]
where with $R(r)$ we denote the amount of bits received by the user at $r$ during the time that its serving bases station dedicates to it. We recall that $(1-\eta)$ is the GPS weight dedicated to information transfer to an IoT UE. This brings to $w_d=\delta_d(1-\eta)$.
For the PS cases, we proceed in the same manner. We write down the ratio between the amount of bits received from the serving base station, using the Shannon capacity formula: 
\[
\frac{w_d\log_{2}\left(1 + \frac{P G r^{-\alpha}}{N_0+I(r,k)}\right)}{\log_{2}\left(1 + \frac{(1-\nu) P G r^{-\alpha}}{N_0+(1-\nu)I(r,k)}\right)}
\]
However such a ratio is a function of $r$, and of the interference experienced by the specific considered user. Thus, we consider the mean user, and set $w_d$ in such a way as to satisfy the equality
\[
\delta_d=\frac{w_d\log_{2}\left(1 + \frac{P G \bar{r}^{-\alpha}}{N_0+\bar{I}(\bar{r},k)}\right)}{\log_{2}\left(1 + \frac{(1-\nu) P G \bar{r}^{-\alpha}}{N_0+(1-\nu)\bar{I}(\bar{r},k)}\right)}
\]
where we approximate the mean interference perceived by the user with the expression in \thref{th:BE_tau}.
\end{proof}

%% file: 05.Optimization.tex
\section{Energy-Optimal Network Configuration}
\label{sec:optimizing}

One of the main open issues in SWIPT networks is to determine, as a function of the main system characteristics as well as of the energy and traffic demands, how the main network parameters should be tuned in order to optimize the energy consumed by the system.
To this end, in this section we elaborate the formulation of the optimization problem which provides, for a given BS energy model, as well as for a given user mean density, the energy optimal BS transmit power, the optimal density of active BSs, as well as the optimal amount of BS time dedicated to power transfer, which satisfy the specified performance constraints.\\  
The objective function is given by the mean power per unit surface consumed by the network, as a function of the BS density, of the BS transmit power, and of the time/power switch ratio. It is thus derived as the product of the power consumed by a BS (expressed in Equation \eref{eq:energymodel}) and of the mean BS density. Note that, by the definition of per bit delay, $U_d=\frac{\bar{\tau}_d(\eta,P,\lambda_b)}{\tau^0_d}$.
The optimization problem which can be formulated with such an objective function clearly depends also on the EH operating mode. In the TS case, we have the following optimization problem formulation:
\begin{problem}{(\textbf{TS operating mode}) }\label{Prob:Opt_1}
\begin{equation}\label{eq:objective_function}
\underset{P,\lambda_b,\eta}{\text{minimize}}\
 \lambda_{b}\left[q_1 +\frac{\bar{\tau}_d(\eta,P,\lambda_b)}{\tau^0_d}\left(q_2+q_3(P-P_{min})\right)\right]
\end{equation}
\vspace{-0.1in}
Subject to:
\vspace{0.1in}
\begin{center}
\begin{varwidth}{\textwidth}
\begin{enumerate}[label=(C\arabic*)]
    \item $ \frac{\bar{\tau}_d(\eta,P,\lambda_b)}{\tau^0_d}\leq 1,\quad \frac{\bar{\tau}_u(\eta,P,\lambda_b)}{\tau^0_u}\leq 1$	 \label{eq:constraint1_downlink}
    \item $ P_{\text{min}}\leq P \leq P_{\text{max}}$ \label{eq:constraint3}
    \item $ 0\leq \eta \leq 1$ \label{eq:constraint4}
    \item $ CDF_h(h_0,\eta,P)\leq \mu$ \label{eq:constraint_charging0}
    \item $ 0\leq \lambda_b \leq \lambda_{b,\text{max}}$ \label{eq:constraint_charging1}
    \end{enumerate}
\end{varwidth}
\end{center}
\vspace{0.1in}
\end{problem}
Constraints C1 require the mean BS utilization in downlink and uplink to be smaller than one. The expression of the ideal per-bit delay is the one from \thref{th:BE_tau}.
Constraints C2 and C3 define the range of acceptable values for the BS transmit power and the time split ratio, respectively. Constraint C4 requires that the probability for an IoT device to harvest less power than the minimum required to operate (denoted as $h_0$, and accounting for the amount of energy required by the device during a whole on-off cycle) be less than a given maximum acceptable value $\mu$. 
The expression of the cumulative distribution function is the one derived in \thref{th:BE_tau_stdev}.
Finally, constraint C5  derives from practical limitations to the maximum density of BS deployments in realistic urban settings, which in any realistic scenario cannot exceed a maximum value. 
\\
\gianluca{In the power splitting EH operating modes (both static and dynamic), the resulting optimization problem takes a very similar structure. Constraint C3 is replaced by
\begin{equation}
\begin{aligned}
\mathrm{(C3a)} \qquad & 0\leq \nu\leq 1
\end{aligned}    
\end{equation}
The other constraints and the objective function remain the same, except that the formulas for $CDF_h$ and $\tau_d$ are those specific to the SPS and DPS modes.}


In all formulations, however, the optimization problem is non-convex and nonlinear, and it cannot be solved efficiently. This is since $\bar{\tau}_d$ (which appears in the expression of both the objective function and in constraints C1) is a nonlinear and non-convex function of $P$, $\lambda_b$, and $\eta$ (or $\mu$). Similar consideration holds also for the expression of $CDF_h(h_0,\eta,P)$ from \thref{th:BE_tau_stdev}, which is nonlinear and non-convex in the three optimization variables.

%% file: 05.1.optimization_problem_study.tex
\subsection{GA Metaheuristic}
One of the possible heuristic approaches to solving Problem 1 is based on an exhaustive grid search and on defining a discrete set of values (e.g., based on practical considerations regarding the maximal accuracy with which they can be varied in practical scenarios) for each of the three decision variables. However, such a brute-force approach becomes quickly unfeasible when the decision variables must be quantized with fine granularity. In this cases, a metaheuristic is required~\cite{rajwar2023exhaustive}. Among these, those based on Genetic Algorithms (GA) \cite{mitchell1998introduction} are particularly well suited for their ability to avoid getting trapped in local minima by maintaining a population of potential solutions. In addition, they are amenable to parallel implementations, and despite not being tied to a specific problem structure, they can be adapted to the specific problem for better efficiency. \biagio{Indeed, as represented in Algorithm \ref{algo:ga}, the GA considers a population of $n$ chromosomes}, each constituting the binary representation of a triplet of values for the three decision variables, is evaluated by assigning the value associated with a fitness function. Such a value is a score indicating the optimality of the solution, and it is used to feed the archive of the best solutions found at a given point in time by the heuristic. This is implemented by picking up the dominant one (\emph{i.e.}, the chromosome with the highest fitness value). Then, a new population is generated using proper genetic operators. \biagio{For the selection, we leveraged classical stochastic uniform selection  \cite{chipperfield1995matlab}, which preserves diversity in the population and helps avoiding premature convergence to suboptimal solutions. Then, we utilize Gaussian Mutation \cite{489178} as the default mutation operator, while the Crossover Heuristic \cite{chipperfield1995matlab} mechanism generates offsprings positioned along the line connecting the two parents at a distance of $ratio$, favoring the direction away from the parent with lower fitness.}
\begin{algorithm}[t]
\caption{Genetic Algorithm for Solving Problem 1}
\begin{algorithmic}[1]
\State Initialize the population $W$ with $n$ chromosomes, where each chromosome represents a binary-encoded triplet $(P, \lambda_b, \eta)$ corresponding to the decision variables;
\Repeat
    \State Evaluate the fitness of each chromosome in $W$ using the penalty-based fitness function (Eq. \ref{eq:fitnessfunction});
    \State Select parents for mating based on their fitness values using \textit{stochastic uniform selection};
    \State Generate an offspring by using \textit{Crossover Heuristic}; 
    \State Mutate the offspring using \textit{Gaussian Mutation};
    \State Evaluate the new offspring and integrate them into the population $W$ by replacing the least fit individuals;
\Until{Iterations = = $maxGen$ OR no improvement in the best fitness value for $i$ consecutive iterations.}
\end{algorithmic}
\label{algo:ga}
\end{algorithm}
\biagio{The algorithm iterates for a maximum of $maxGen$ generations or stops earlier if the termination condition is met, defined as no improvement in the best fitness value for $i$ consecutive generations. In this way, guided by the fitness function, the GA realizes an optimistic exploration of the solution space based on the evolutionary operators' output.}

\biagio{To facilitate the convergence of the GA, the fitness function is defined through the penalty method \cite{bazaraa2013nonlinear}, by which, in addition to the cost function of Problem \ref{Prob:Opt_1}, a penalty is introduced as a linear combination of the constraints, in such a way that the fitness function value increases proportionally to the degree by which constraints \ref{eq:constraint1_downlink} and \ref{eq:constraint4} over the utilization and minimum harvested power are violated. This ensures that the resulting metaheuristic can also explore those values of the decision variables that are very close to the border of the feasibility region of Problem 1. The resulting expression for the fitness function is as follows:}
\begin{align}
\label{eq:fitnessfunction}
v(P,\lambda_b,\eta) =& \lambda_{b}\left\{q_1 +\frac{\bar{\tau}_d(\eta,P,\lambda_b)}{\tau^0_d} [q_2+q_3(P-P_{min})]\right\}
+ k_1 \left(\frac{\bar{\tau}_u(\eta,P,\lambda_b)}{\tau^0_u}-1\right)^{+}+\\
-& k_2[CDF_h(h_0,\eta,P)-\mu]^{-}+k_3\left(\frac{\bar{\tau}_d(\eta,P,\lambda_b)}{\tau^0_d}-1\right)^{+} \nonumber
\end{align}
where $(x)^{+}=max(x,0)$, and $(x)^{-}=min(x,0)$. 
\biagio{$k_1$, $k_2$, and $k_3$ are nonnegative weights, assigned empirically to facilitate the efficient exploration of the border of the feasibility region. In our experiments, we have verified that by choosing for the three weights a value of the same order of magnitude as the maximum value of the objective function in Problem 1, i.e. equal to $\lambda_{b,max}[q_1+q_2+q_3(P_{max}-P_{min})]$ is sufficient to significantly enhance both convergence consistency and solution quality with respect to the approach without penalty, effectively accelerating the convergence of the algorithm.}

%% file: 06.Numerical.tex
\section{Numerical Results}
\label{sec:numerical}
In this section, we validate numerically our analytical results, and we investigate the properties of different strategies for energy efficiency in SWIPT networks.
\subsection{Setup}
We assume base stations work in the $1.5$ GHz band and use a bandwidth of $50$ MHz, compatible with 4G+ standards.
Unless otherwise indicated, we assume a percentage of IoT devices equal to $80\%$ of the total number of UEs (representative of many present-day scenarios), a transmit power equal to $0.2$ W for both IoT and BB UEs, and a frequency reuse factor of $3$. Moreover, by default, we assume a beamforming gain equal to $10$, constant over the whole main lobe aperture (which we assume to be equal to $45$ degrees), and a path loss exponent equal to $3$, typical of urban areas. 
For linear energy harvesting, we assume a conversion efficiency of $0.9$, with no lower/upper threshold. 
We assume the deployed base stations to be of the macro type, with a transmit power that varies between $1$ and $11$ W. Unless otherwise stated, we set a target mean per-bit delay in downlink for BB (resp. IoT) UEs equal to $10^{-5}$ s (resp. $10^{-3}$ s), and in uplink equal to $10^{-4}$ s for all UEs,  (e.g., typical of IoT systems for environmental monitoring \cite{lu_energy_2021}). We consider the user density to vary from $10^{-4}$ users per $m^2$ (modeling settings with a high share of BB users) to $10^{-1}$ users per $m^2$ (modeling scenarios with crowds of BB UEs and high-density IoT deployments). Unless otherwise specified, we assume IoT UEs to be active all the time ($\phi = 1$).
We set to $5\%$ the maximum acceptable share of IoT users who are not able to harvest the target minimum energy.

The parameters of the BS energy model are chosen to fit two different types of BSs. The first type (labeled LLP – low load proportionality) reflects the behavior of the majority of current stand-alone BSs, and is characterized by a $27\%$ load proportionality (with $q_1 = 1100$, $q_2 = 100$, and $q_3 = 30$). Conversely, the high load proportionality (HLP) BS type (with $q_1 = 482.3$, $q_2 = 48.23$, and $q_3 = 144.69$) corresponds to a $75\%$ load proportionality, achievable, e.g., through time-domain duty-cycling at the sub-system level, i.e., through micro-sleep techniques involving modules of the BS or of the BBU in cloud-RAN designs \cite{lin2021data}. For the sake of comparison, these parameters were chosen to fit a per-BS maximum consumed power of $1500$ W, typical of stand-alone macro BSs \cite{ge2017energy}.
The value of BS density has been varied with a granularity of $10^{-4}$ $m^{-2}$, to guarantee good accuracy of the GA search process. We assume the mean number of users per base station to be lower bounded by $5$. This models the simple energy-saving strategy common among MNOs, which switch off those BSs that serve very few users to no user at all, as they represent a very high energy cost per user and a small benefit for performance.    
The size $n$ of the initial population in the GA algorithm was set to $100$ chromosomes. This choice proved to be a good compromise between computational load and convergence speed.
In the fitness function, parameters $k_1$, $k_2$ and $k_3$ were set to $1,000$, $5,000$, and $1,000$, respectively, as these values proved appropriate to enable our GA approach to explore effectively the borders of the feasibility region. 
\biagio{The termination condition of the GA was set based on the convergence of the fitness value. Specifically, we assume that convergence is reached when the geometric average of the relative change in the value of the spread over ($ i = 10$) generations is less than $10^{-6}$, and the final spread is less than the mean spread over the past $i$ generations. We also set a generation limit $m$ to 150 to prevent potential infinite iterations. This value was chosen based on an empirical evaluation of the typical number of iterations required across different starting points. Regarding the Crossover Heuristic, we choose a $ratio = 1.9$, to have a good distance from parents, and prevent the algorithm from easily getting stuck at a local minimum.} 

\subsection{Computational Complexity}
\biagio{
As for the time complexity of the proposed GA, in the worst case scenario, if $n$ iterations are required per each population of $m$ chromosomes, and if with $T_{sim}$ we denote the time to compute the value of the function modeled by the chromosome and to perform the estimation of the fitness value, the time complexity of our algorithm is $\mathcal{O}(mnT_{sim}$). In our setup, consisting in an Intel Core i5 6-core - used in parallel - processor and 16 GB 2667 MHz DDR4 RAM, $T_{sim}$ has been always less than 2s.
However, we have empirically observed that in $x$ tries of the GA, out of the number of points evaluated for all the six configurations analyzed, our termination condition always brought the GA to converge before the maximum number of iterations. This confirmed the validity of our heuristic choice of the three constants $k_1$, $k_2$, and $k_3$, 
as it is hard to model analytically the mutual relationships between these constants and their impact on the resulting value assumed by the overall function.}
\subsection{Linear Energy Harvesting}
To assess the accuracy of our results, we simulated the system for several values of user and BS densities, in the linear EH case. 
For each of the three power harvesting configurations, and over different values of user density, \fref{fig:power_consumed_validated} shows the values of power per $km^2$ consumed by the network at the optima derived by our GA, as well as those derived by simulating the system in those optima.\\
As the figure shows, the power consumed by EH IoT users has a strong impact on the total power consumed by the network, irrespective of the degree of energy proportionality of BSs. Indeed, at low user densities, consumed power more than doubles when the target minimum harvested power passes from $1$ to $6$ mW. At high user densities instead, the difference is smaller because the power consumed to provide information transfer dominates the overall energy footprint of the network. This is due also to the growing inefficiency of information transfer at high user densities, due to rising interference levels.\\  
\fref{fig:power_consumed_validated}
 suggests also that our modeling approach is extremely accurate across very diverse system configurations. Indeed, the markers reporting simulation results overlap the curves obtained with our analytical model. In particular, very good accuracy is achieved even for low values of user density, for which the dense IoT regime assumption of Theorem 2 does not hold.\\
\begin{figure*}[t!]
\centering
\subfloat[Time Switching]{%
\includegraphics[width=0.33\columnwidth, height=1.5in]{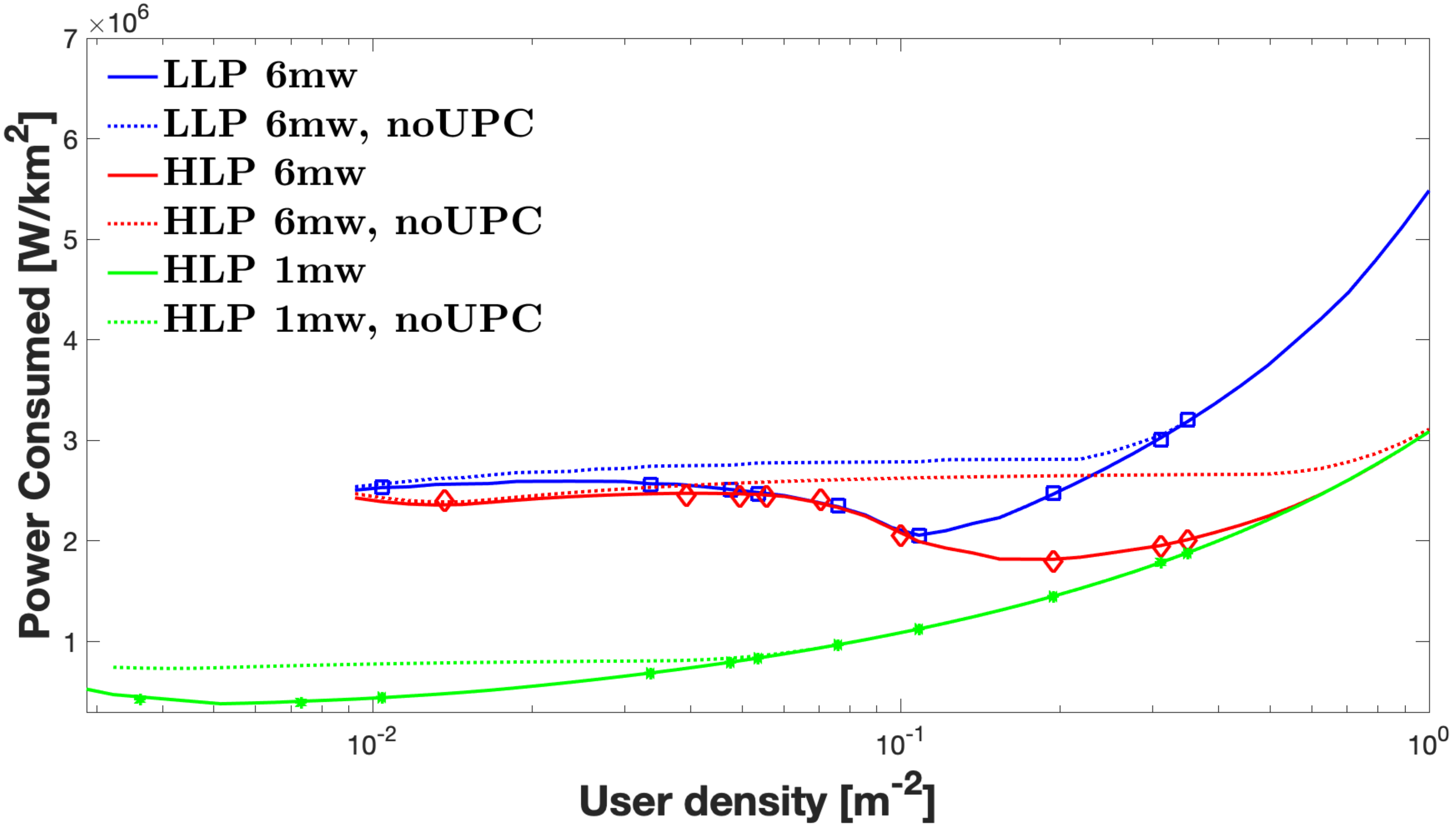}\label{fig:power_consumed_validated_ts}%
}\hspace{-0.05 in}
\subfloat[Dynamic Power Splitting]{%
\includegraphics[width=0.33\columnwidth, height=1.5in]{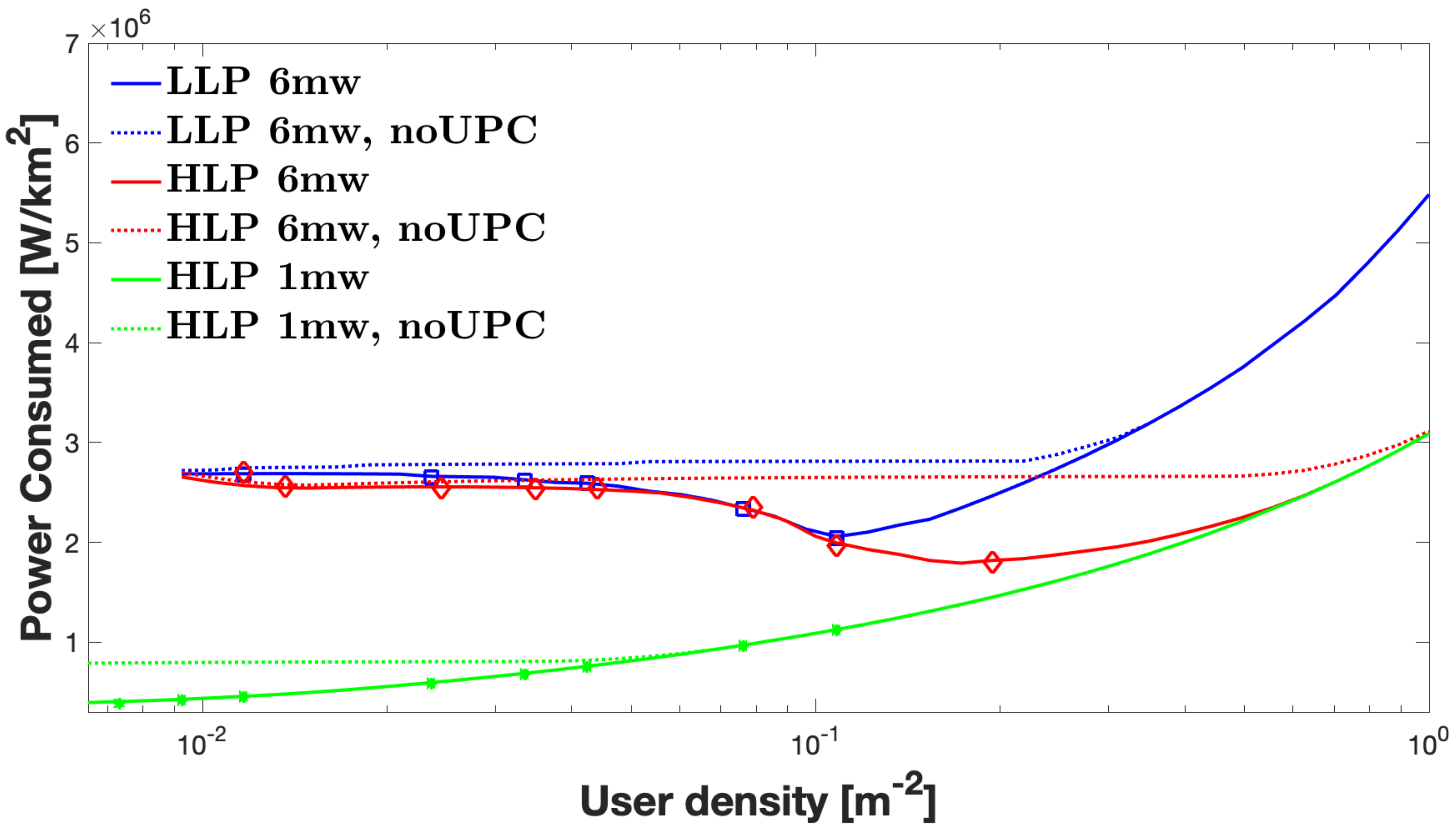}\label{fig:power_consumed_validated_dps}%
}\hspace{-0.05 in}
\subfloat[Static Power Splitting]{%
\includegraphics[width=0.33\columnwidth, height=1.5in]{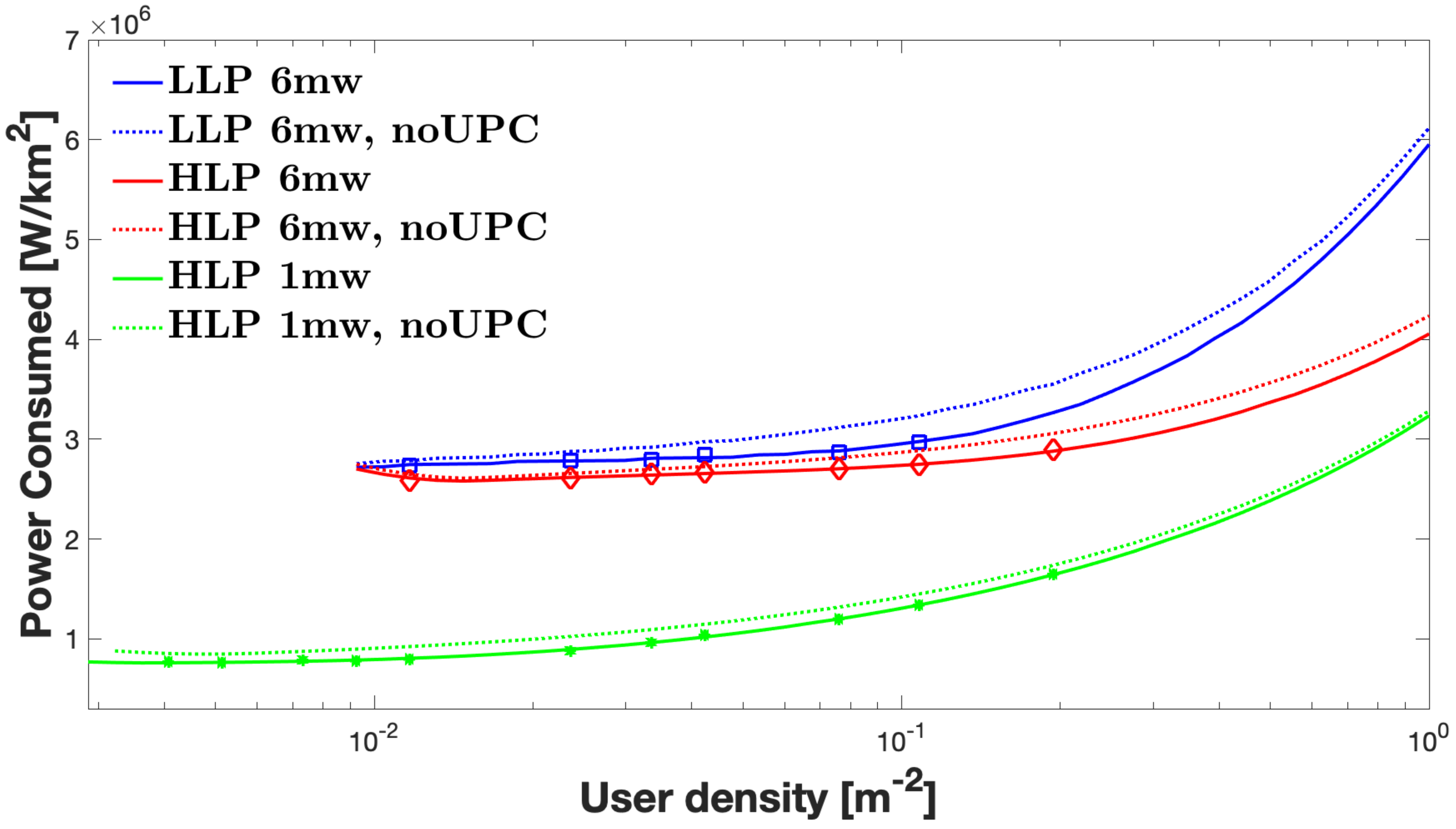}\label{fig:power_consumed_validated_sps}%
}
\caption{Power per $km^2$ consumed by the network at the optimum vs user density, for different target minimum harvested power and IoT UE receiver configurations. Markers denote values from simulations,  derived with a $95\%$ confidence interval of at most $8\%$.\normalsize}
\label{fig:power_consumed_validated}
\end{figure*}
These results suggest that, in sharp contrast to RANs delivering only connectivity, as the density of IoT users increases, their contribution to the service capacity of the whole SWIPT network (in the form of power delivered to other IoT users) increases too. The effect is visible in all configurations when comparing the setting in which the energy harvested originates only from BS transmissions (the "no UPC" curves in \fref{fig:power_consumed_validated}),
with those in which users can harvest energy also from transmissions of other users.
In all configurations, we find that as user density increases, its effects manifest as a decrease in the power consumed by the network. Such a decrease continues until a minimum is reached, after which the power consumed increases again with increasing user density. Such behavior is quite surprising, but is the result of the interplay between, on the one side, the detrimental effect of rising interference on communications (which decreases the efficiency with which network resources are used), and on the other, the beneficial effect of the increase in user transmissions on the amount of power harvested by IoT users. When increasing user density from very low values,  these two contrasting effects cause the overall power consumed by the network to decrease, before increasing again due to the effects of high levels of interference on user-perceived QoS for data. The impact of user contribution to network service capacity is less marked when the target minimum service capacity is low (e.g., in the $1$ mW plots), due to the low impact of power delivery on the overall consumed power of the network, and it increases for higher target minimum harvested power. Such impact is limited also in the SPS mode, due to the lower efficiency of passive EH in that configuration.\\  
 \begin{figure}[t!]
\centering
\begin{minipage}{.45\textwidth}
  \centering
\includegraphics[width=\columnwidth]{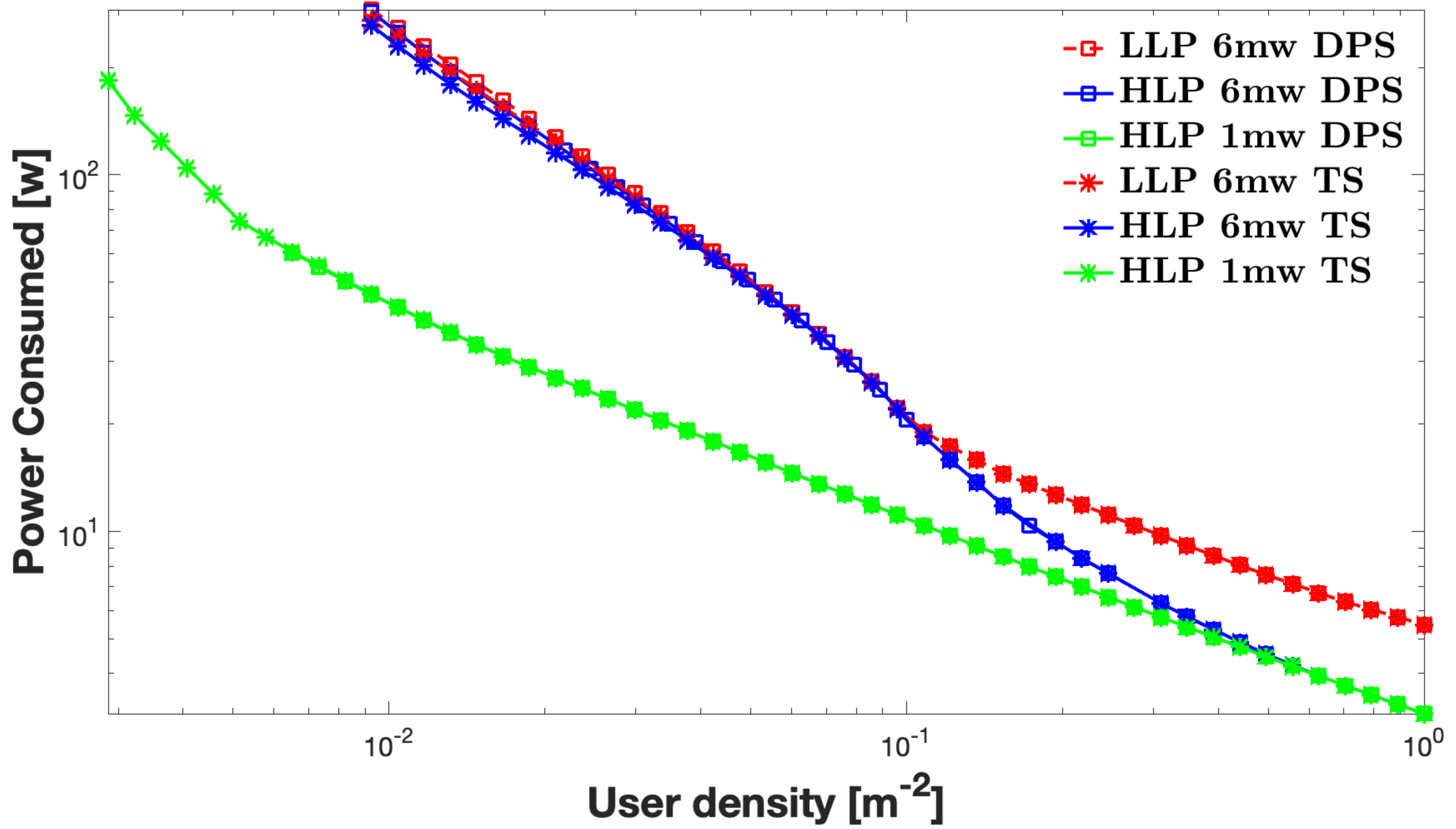}
\caption{Power consumed by the network for each user (broadband or IoT) vs. user density, for different target minimum harvested power and IoT UE receiver configurations.\normalsize}
\label{fig:powerconsumed_peruser}%
\end{minipage}%
\hspace{0.2in}
\begin{minipage}{.45\textwidth}
  \centering
  \includegraphics[width=\columnwidth]{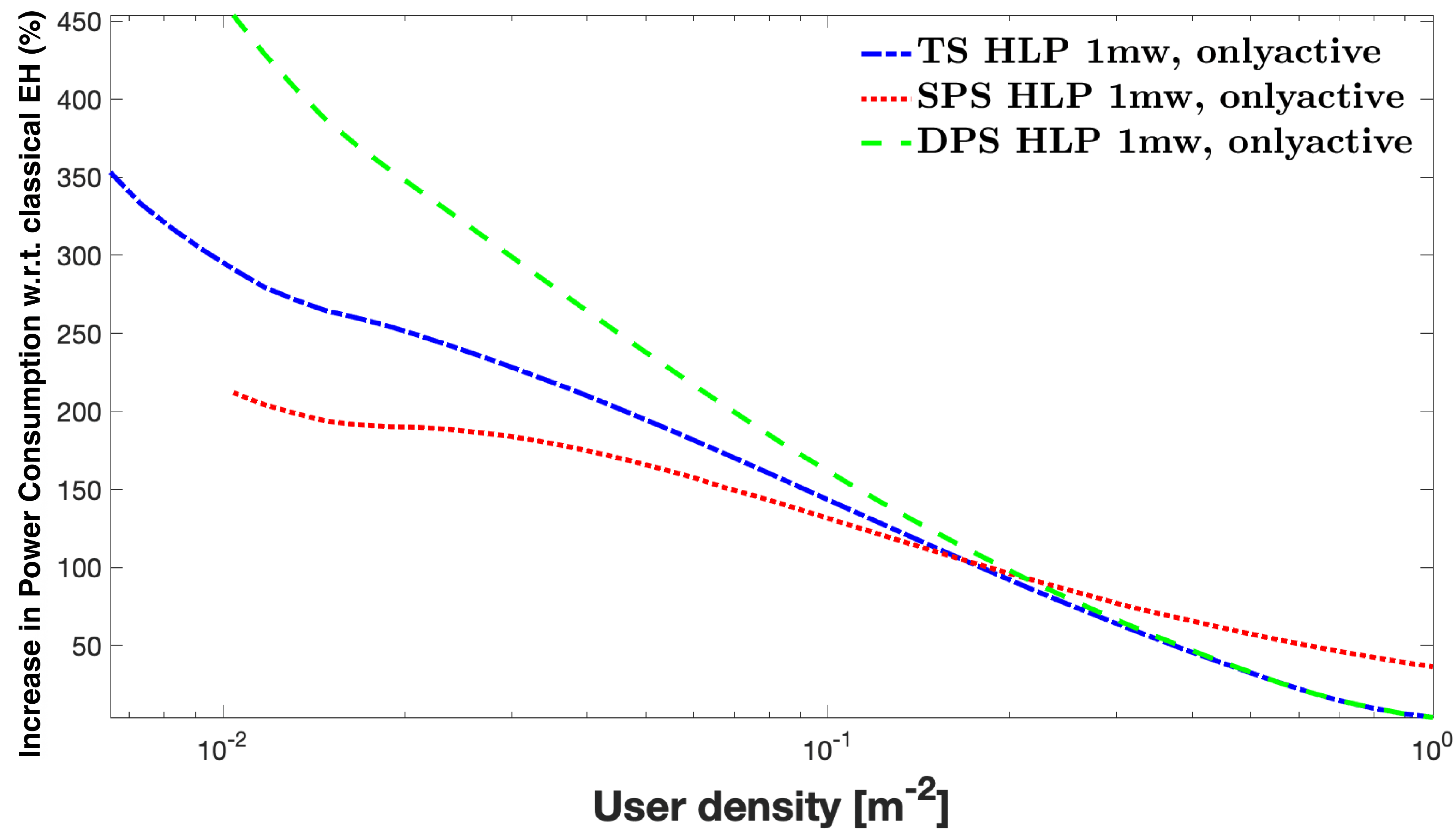}
\caption{ Increase in the power per $km^2$ consumed by the network with respect to linear EH, as a function of user density, for HLP $1$ mW target minimum harvested power, and for different IoT UE receiver configurations.\normalsize}
\label{fig:comparison_onlyactive}%
\end{minipage}
\end{figure}
These features of optimal power consumption suggest that (possibly dynamic) configuration tuning holds the potential to play a key role in minimizing the energy footprint of a SWIPT network. Indeed, these features suggest that by taking advantage of user-provided service capacity, it is possible to make the network increasingly more energy efficient when the demand (of both energy and communication) it must satisfy grows. This is suggested also by \fref{fig:powerconsumed_peruser}, which shows that for all configurations and target minimum harvested power, the network energy efficiency increases when user density increases. Here too, when the user-provided service capacity 
dominates the overall energy footprint of the network, the decrease of per-user energy consumed is faster, while it tends to slow down for those user densities for which the effects of interference dominate. Quite surprisingly, for the same target minimum harvested power, when user-provided service capacity dominates, the power per user consumed by the network does not depend on the amount of load proportionality of the BS energy model (red and blue curves are very close). It is only when interference dominates that the higher BS densities required to satisfy the target QoS for information transfer bring an increase in the impact of the fixed energy costs of BSs on the overall energy footprint of the network.\\  
The importance of passive energy harvesting in SWIPT networks is also visible from \fref{fig:comparison_onlyactive}, which shows the increase in the network consumption when only active charging from the serving BS is available. Indeed, when passive EH is not available (e.g., due to limitations in the HW architecture of devices) the energy consumed by the network can be several times larger than in the case in which passive EH is available.\\
Another noteworthy feature of a SWIPT network emerging from \fref{fig:power_consumed_validated} is that for each configuration of EH receiver architecture and target minimum harvested power, there is a value of user density below which no feasible solution of Problem 1 exists. This is due to the lower bound to the mean number of users per base stations (equal to $5$ in our evaluations). Indeed, as with increasing distance from the serving base station the power received by users decays much faster than channel capacity, guaranteeing a target QoS for harvested power at low user densities requires ever increasing BS densities and thus very low mean number of users per BS, up to the point of economic unfeasibility. 
Note however that this is relative to a scenario in which users are distributed uniformly at random in space. In realistic settings, this implies that deployments of small amounts of EH IoT devices are of course feasible, but they should be, at least locally, dense enough to be able to deliver power with the target QoS while satisfying the system constraints.
%
\begin{figure*}[t!]
\centering
\subfloat[Time Switching]{%
\includegraphics[width=0.33\columnwidth]{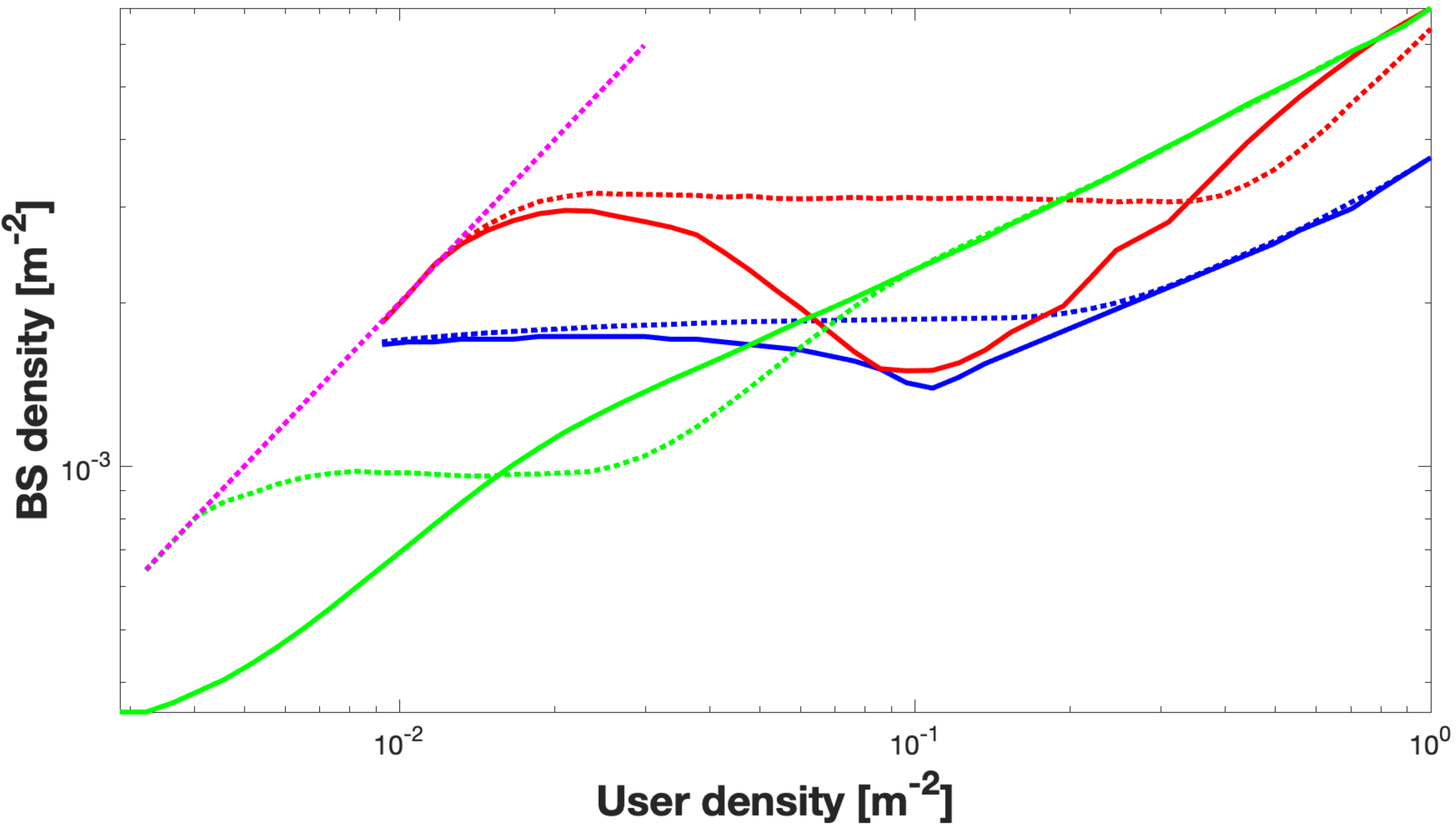}\label{fig:optimal_bsdensity_ts}%
}\hspace{-0.05 in}
\subfloat[Dynamic Power Splitting]{%
  \includegraphics[width=0.33\columnwidth]{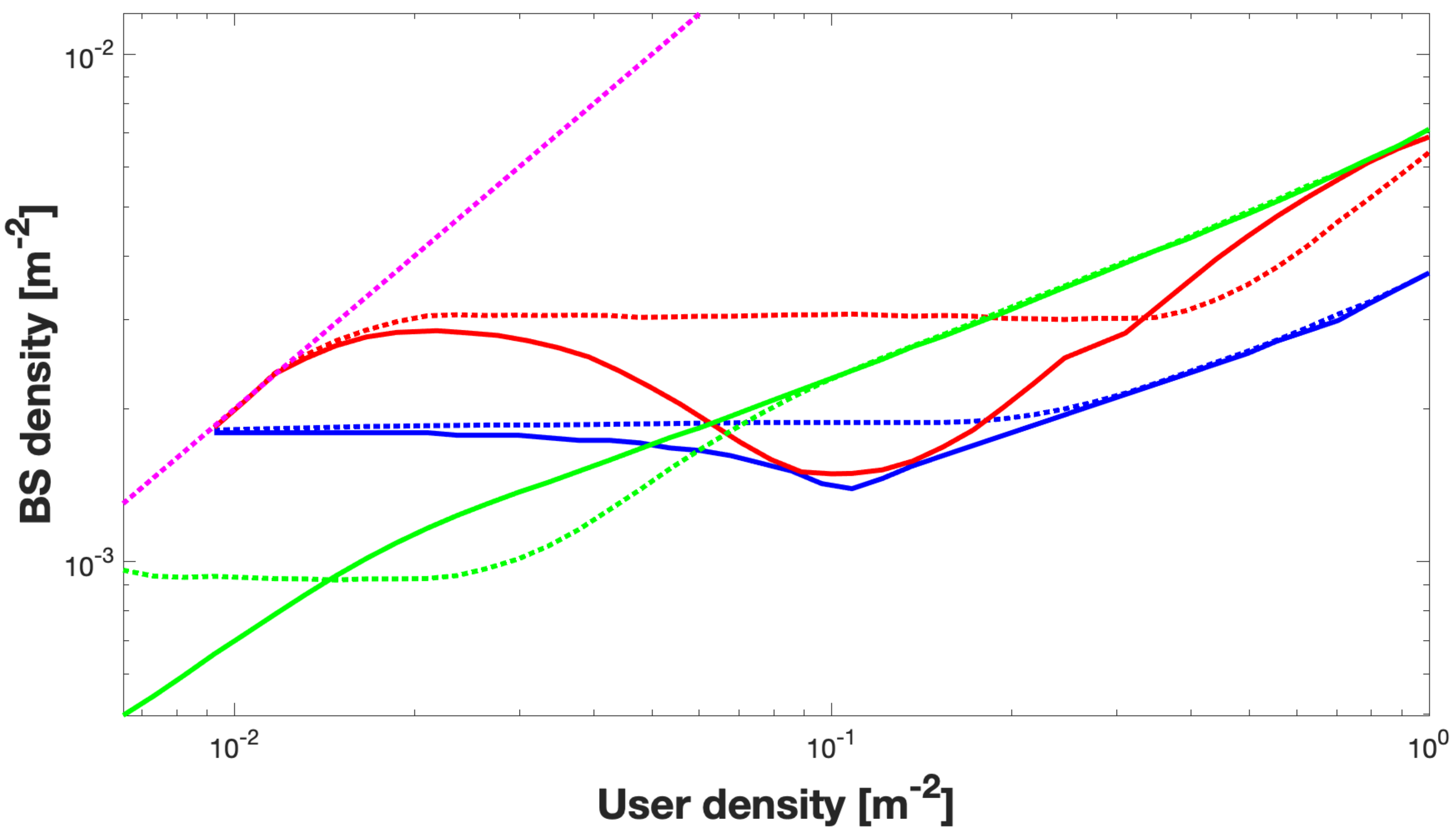}\label{fig:optimal_bsdensity_dps}%
}\hspace{-0.05 in}
\subfloat[Static Power Splitting]{%
  \includegraphics[width=0.33\columnwidth]{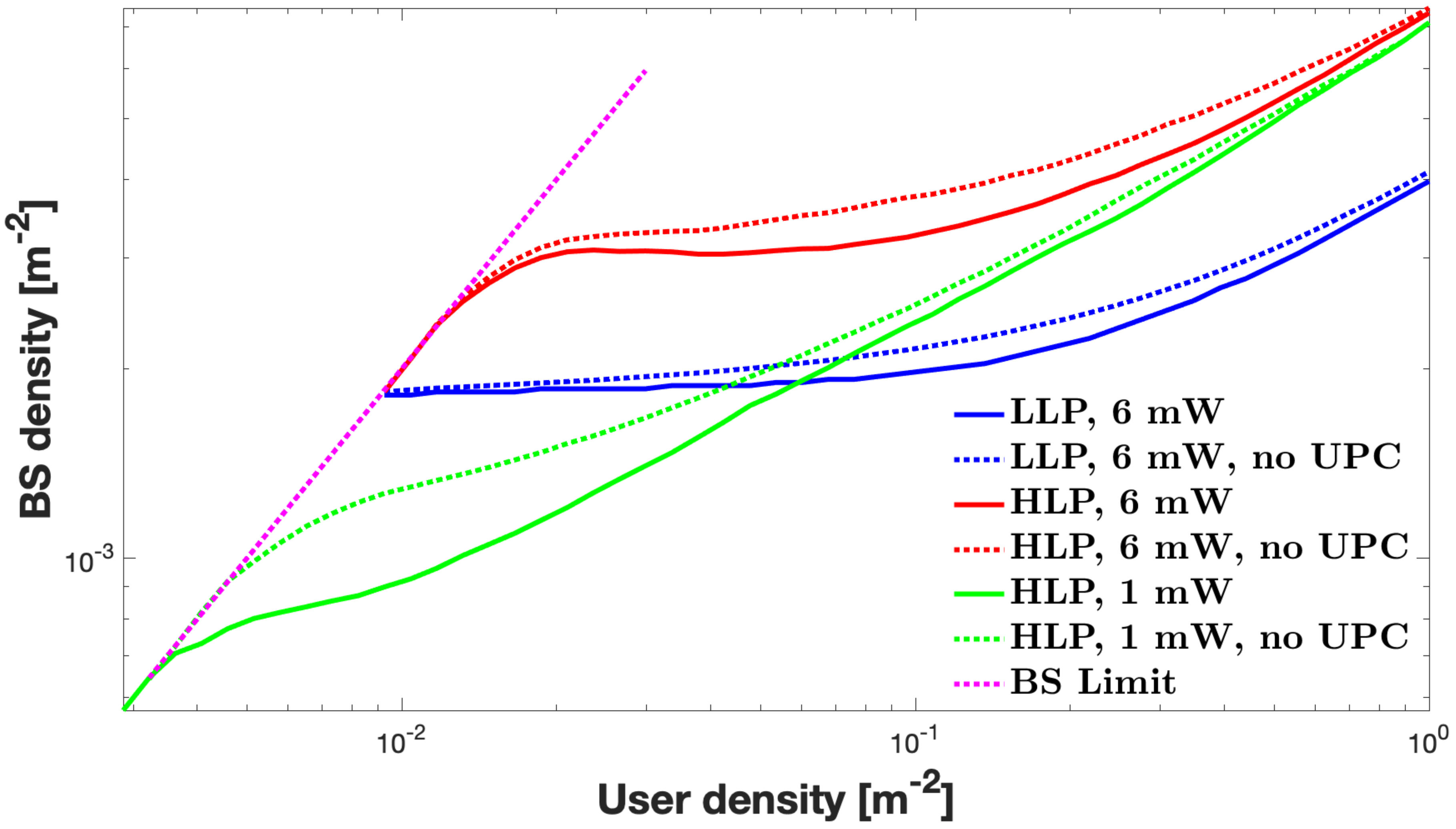}\label{fig:optimal_bsdensity_sps}%
}
\caption{Optimal BS density for different target minimum harvested power and IoT UE receiver configurations.}
\label{fig:optimal_bsdensity}

\centering
\subfloat[Time Switching]{%
\includegraphics[width=0.33\columnwidth]{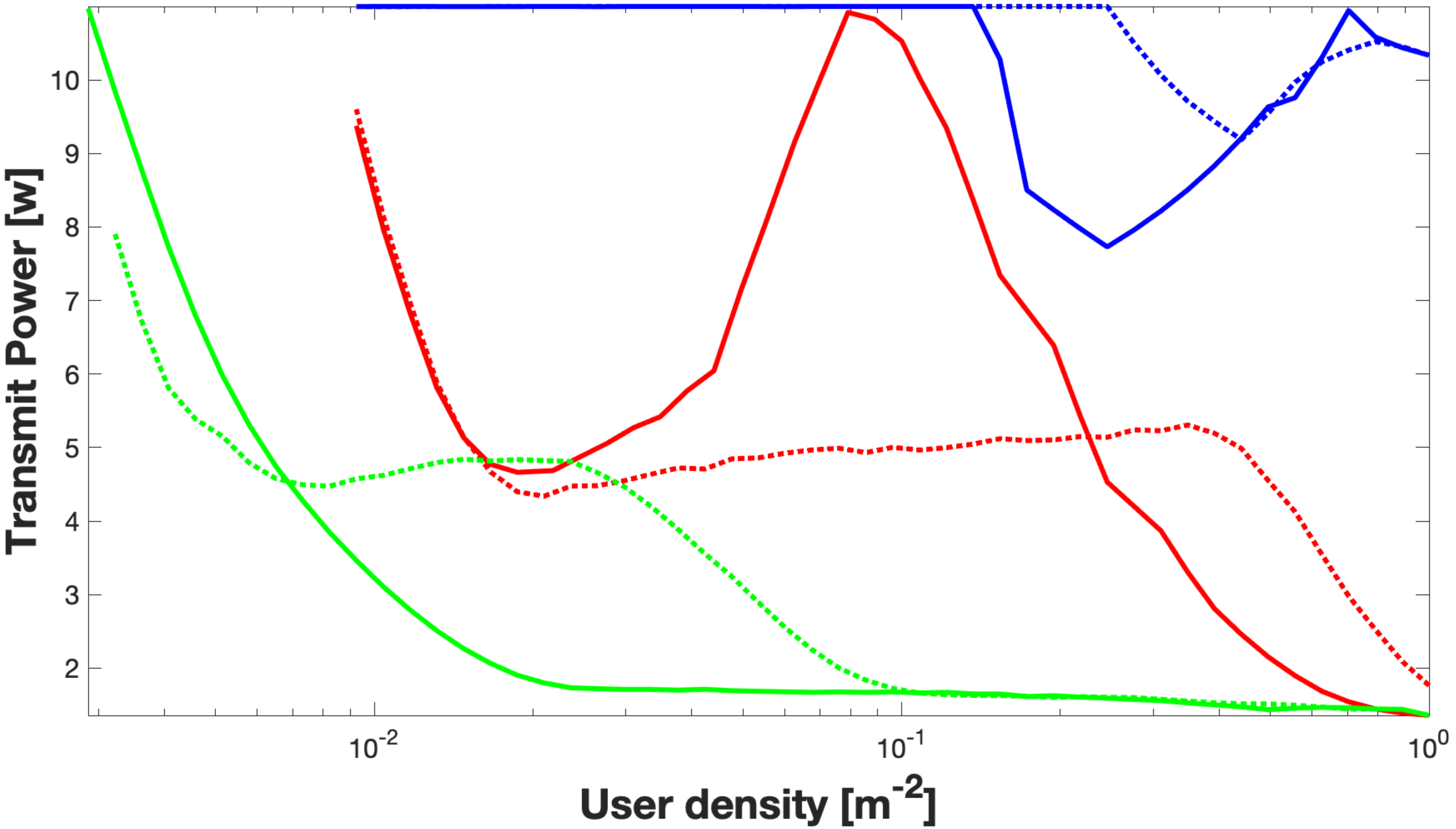}\label{fig:optimal_txpower_ts}%
}\hspace{-0.05 in}
\subfloat[Dynamic Power Splitting]{%
  \includegraphics[width=0.33\columnwidth]{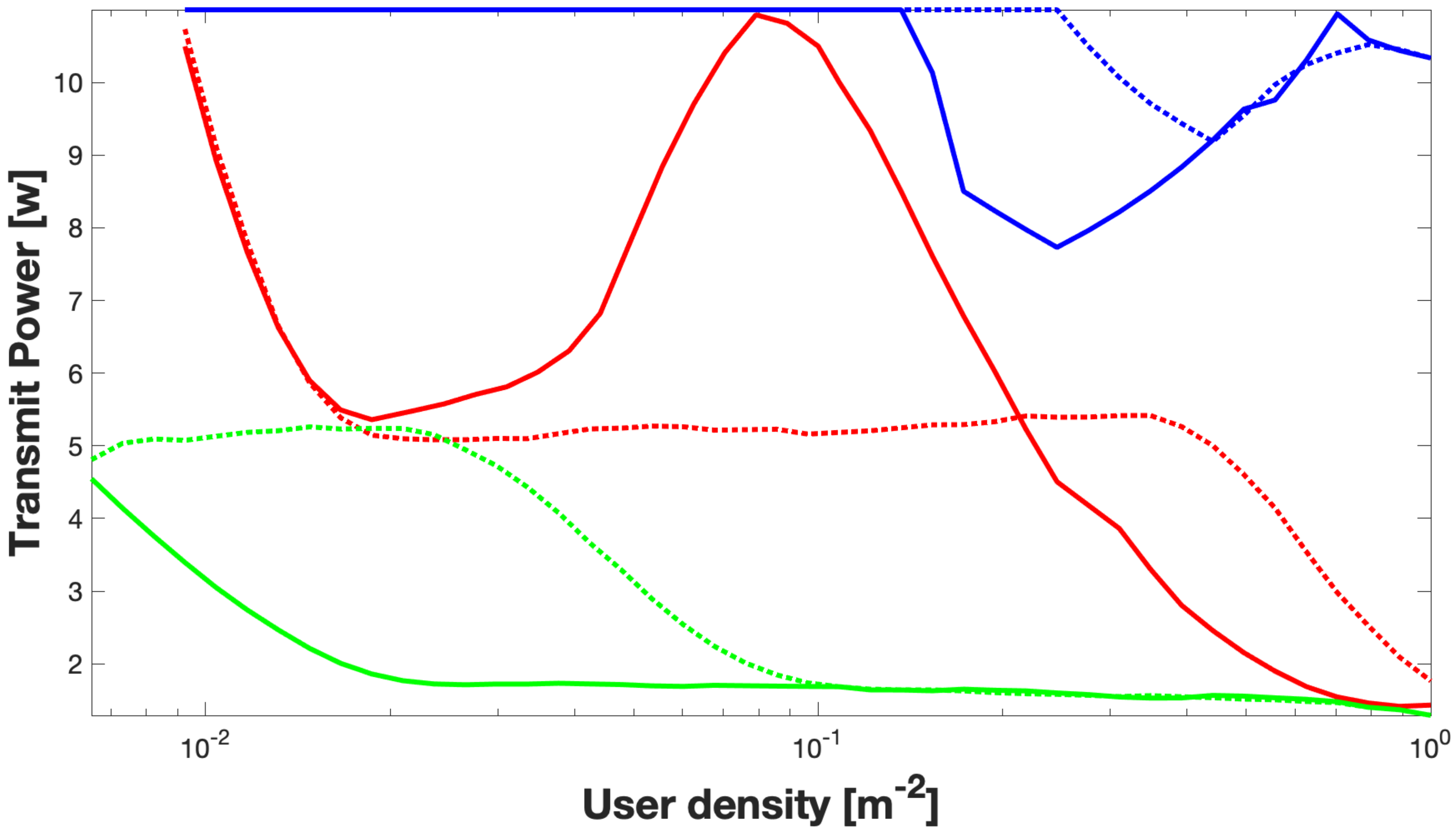}\label{fig:optimal_txpower_dps}%
}\hspace{-0.05 in}
\subfloat[Static Power Splitting]{%
  \includegraphics[width=0.33\columnwidth]{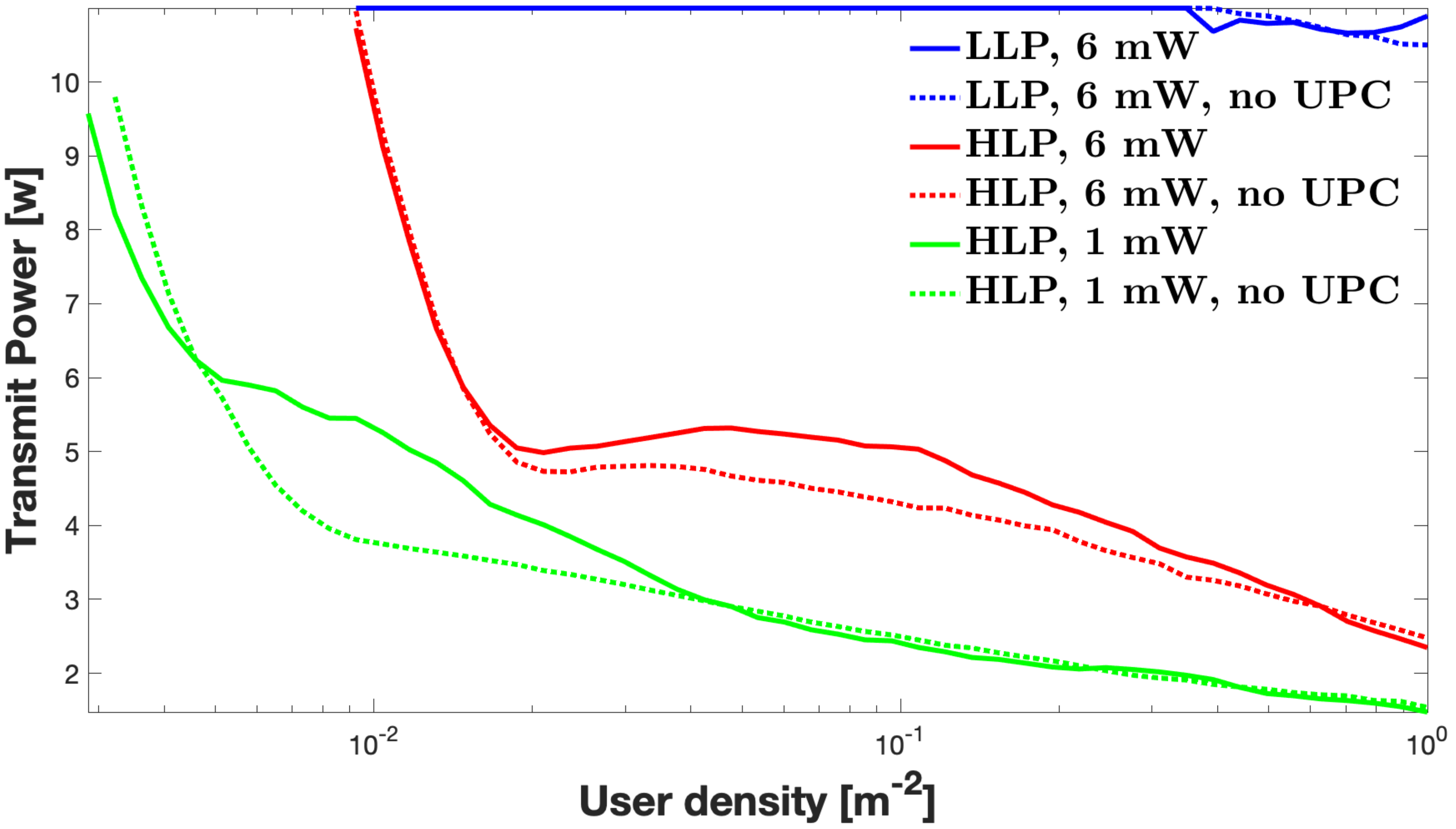}\label{fig:optimal_txpower_sps}%
}
\caption{Optimal transmit power for different target minimum harvested power and IoT UE receiver configurations.}
\label{fig:optimal_txpower}

\centering
\subfloat[Time Switching]{%
\includegraphics[width=0.33\columnwidth]{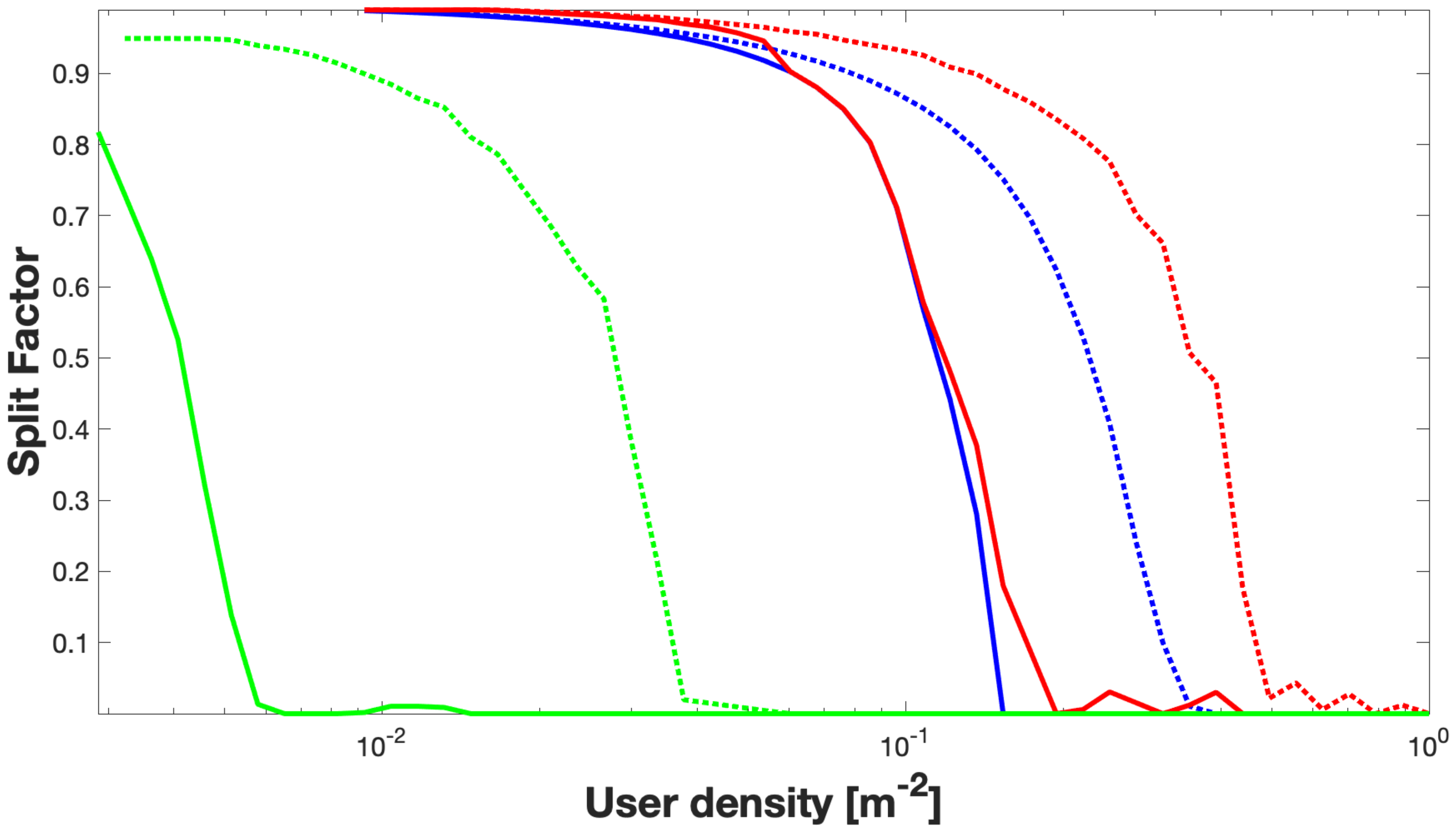}\label{fig:optimal_factor_ts}%
}\hspace{-0.05 in}
\subfloat[Dynamic Power Splitting]{%
  \includegraphics[width=0.33\columnwidth]{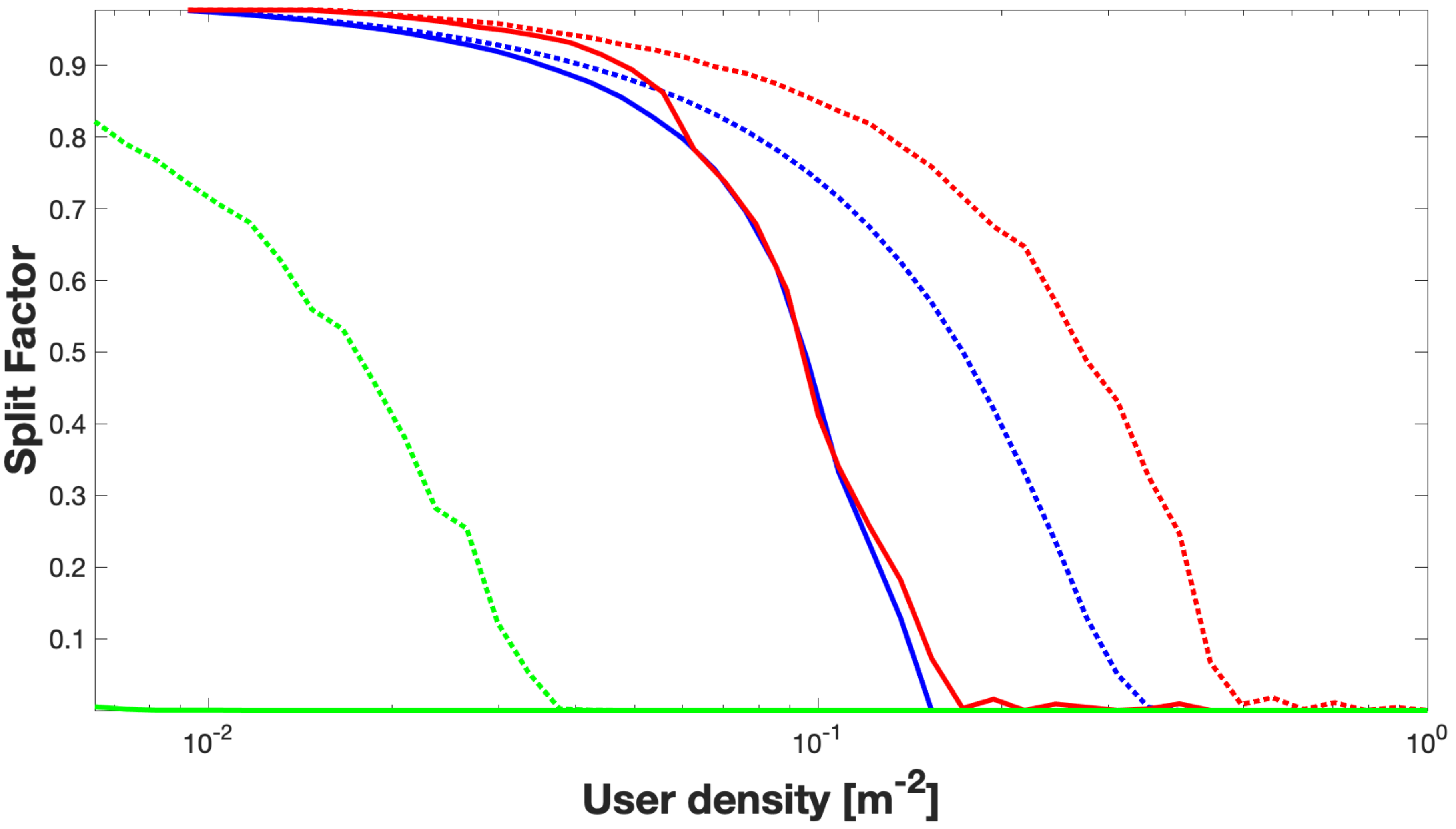}\label{fig:optimal_factor_dps}%
}\hspace{-0.05 in}
\subfloat[Static Power Splitting]{%
  \includegraphics[width=0.33\columnwidth]{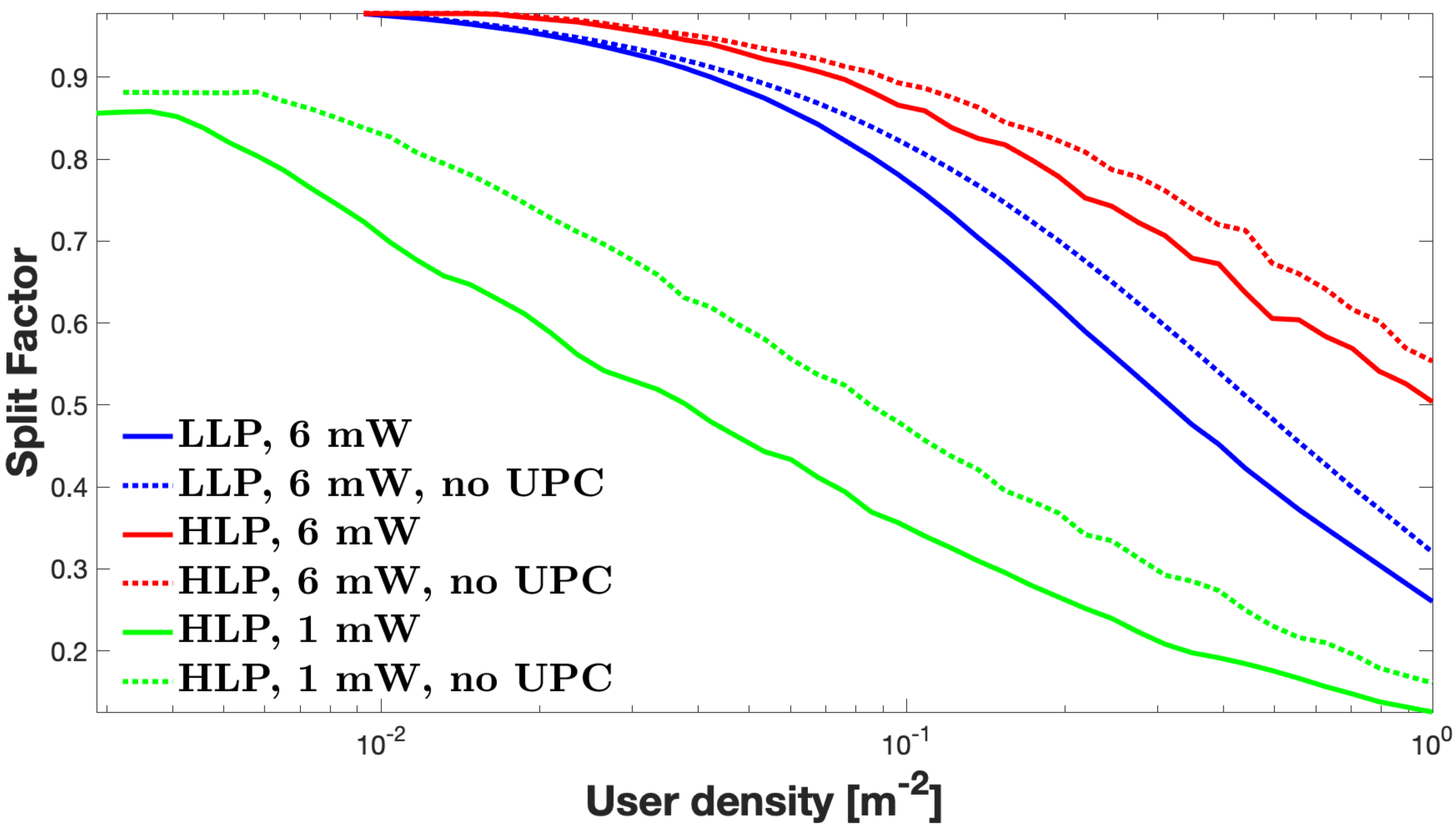}\label{fig:optimal_factor_sps}%
}
\caption{ Optimal time splitting factor (for TS) and power splitting factor (for DPS and SPS) for different target minimum harvested power and IoT UE receiver configurations.}
\label{fig:optimal_factor}
\end{figure*}
In \fref{fig:optimal_bsdensity}, \ref{fig:optimal_txpower}, and \ref{fig:optimal_factor} we plot the optimal values of BS density, of transmit power, and of time split ratio (respectively, power split ratio), for the three EH receiver configurations.
As \fref{fig:optimal_bsdensity} shows, in the TS and DPS cases the evolution of the optimal BS density as a function of user density is characterized by three regimes. First, for very low user densities, the optimal BS density increases with user density, as expected. For higher user densities, the optimal BS density reaches a plateau in the case without user-supplied service capacity, and it decreases in the case in which the full potential of passive EH is available. This feature suggests new and quite surprising patterns of BS sleep modes which, in contrast to those for non-SWIPT RANs, turn off BSs when demand increases. Finally, for even larger user densities, in the interference-dominated regime, the optimal BS density increases again. These three regimes are not present however for all values of target minimum harvested power, due to the minimum user-per-BS limit adopted. These three regimes are not present in the SPS case, in which the lower efficiency of passive EH brings optimal BS densities which increase monotonically with user density.\\
From \fref{fig:optimal_txpower} we see that the optimal transmit power varies with user density in a specular manner to optimal BS density, decreasing when it increases, and vice versa. That is, the energy optimal strategies emerging from our GA algorithm tend to compensate for the thinning (resp. densification) of BS density with cell zooming (resp. shrinking), via transmit power tuning. This is likely due to the fact that, when IoT node density increases, the impact on the system of the inefficiency due to high distance between BS and EH IoT user is stronger, pushing the optimal solution towards larger BS densities.\\
Fig. \ref{fig:optimal_factor} offers some key insights on how tuning split factor contributes to achieving energy-optimal operation of a SWIPT network. For low-load proportional BSs, as well as for configurations with high target minimum harvested power, the bulk of the energy consumed must be supplied by the serving BS, which thus has to allocate the majority of its time to transferring energy to IoT users. As expected, such a share of BS time is larger in scenarios with a larger target minimum harvested power. However, in all configurations, for increasing user densities, the share of BS time dedicated to active charging decreases, thanks to the combined effect of user-supplied service capacity, and higher BS density. Beyond a given user density (which depends on the BS energy model and receiver EH architecture) the SWIPT network effectively stops actively delivering power to IoT users (the split factor gets values close to zero), as passive power supply from ambient sources is sufficient to achieve the target QoS for energy harvesting. Again, note that for the SPS configurations, as harvesting from user transmissions is less efficient, the decrease in the splitting factor with increasing user density is substantially slower than in TS or DPS.\\
Finally, to characterize the effectiveness of our GA algorithm, we have compared the results derived from it with those achieved with a simple grid search heuristic. This latter is based on the discretization of the optimization parameters, and on an exhaustive search over these discrete values. In all tests, over all considered settings and for different values of the discretization step, the difference in results were only due to the chosen discretization step, and could be made arbitrarily small with a finer discretization step. This suggests that, for all the settings we have considered, our GA algorithm proved effective in avoiding the search process to get stuck into local minima, while being substantially more efficient.

\subsection{Impact of harvesting nonlinearities}
A crucial aspect of the energy harvesting process is the presence of nonlinearities in the relation between the output power of the energy harvester circuit at the IoT devices and its input power. Indeed, such aspect may heavily affect the amount of harvested power as a function of the network settings, and thus the energy-optimal configuration of the whole SWIPT network. To investigate its impact on the energy-optimal strategies derived by our GA, in another set of experiments we considered the energy harvesting curve given by Equation \ref{eq:sigmoid}. Following indications from the literature \cite{wang_wirelessly-powertwoway_2017}, we have set the maximum output power $h_{max}$ to $10$ mW, $\chi$ to $274$, 
and the sensitivity threshold $h_s$ to $0.064$ mW. Figure \ref{fig:mapping_function_nonlinear} shows the nonlinear EH curve for the chosen parameters and for two values of $\iota$, which tunes the maximum efficiency of the nonlinear harvesting process. Its values have been chosen not to have a maximum efficiency larger than one, and to approximate the linear harvesting curve which we have considered in the previous experiments. In the figure, we have also plotted the linear harvesting functions whose slope approximates well that of the two sigmoids. As the plot shows, both sigmoids are characterized by a range of input power within which they are well approximated by a linear function (and which we denote henceforth as a "quasilinear regime").
Note that, with the given settings, choosing $\iota=0.9$ (respectively, $1.30$) produces a sigmoid which in the quasilinear regime approximates well a linear model with slope $0.9$ (respectively, $0.77$). 
 \begin{figure}[t!]
\centering
\begin{minipage}{.45\textwidth}
  \centering
\includegraphics[width=\columnwidth]
{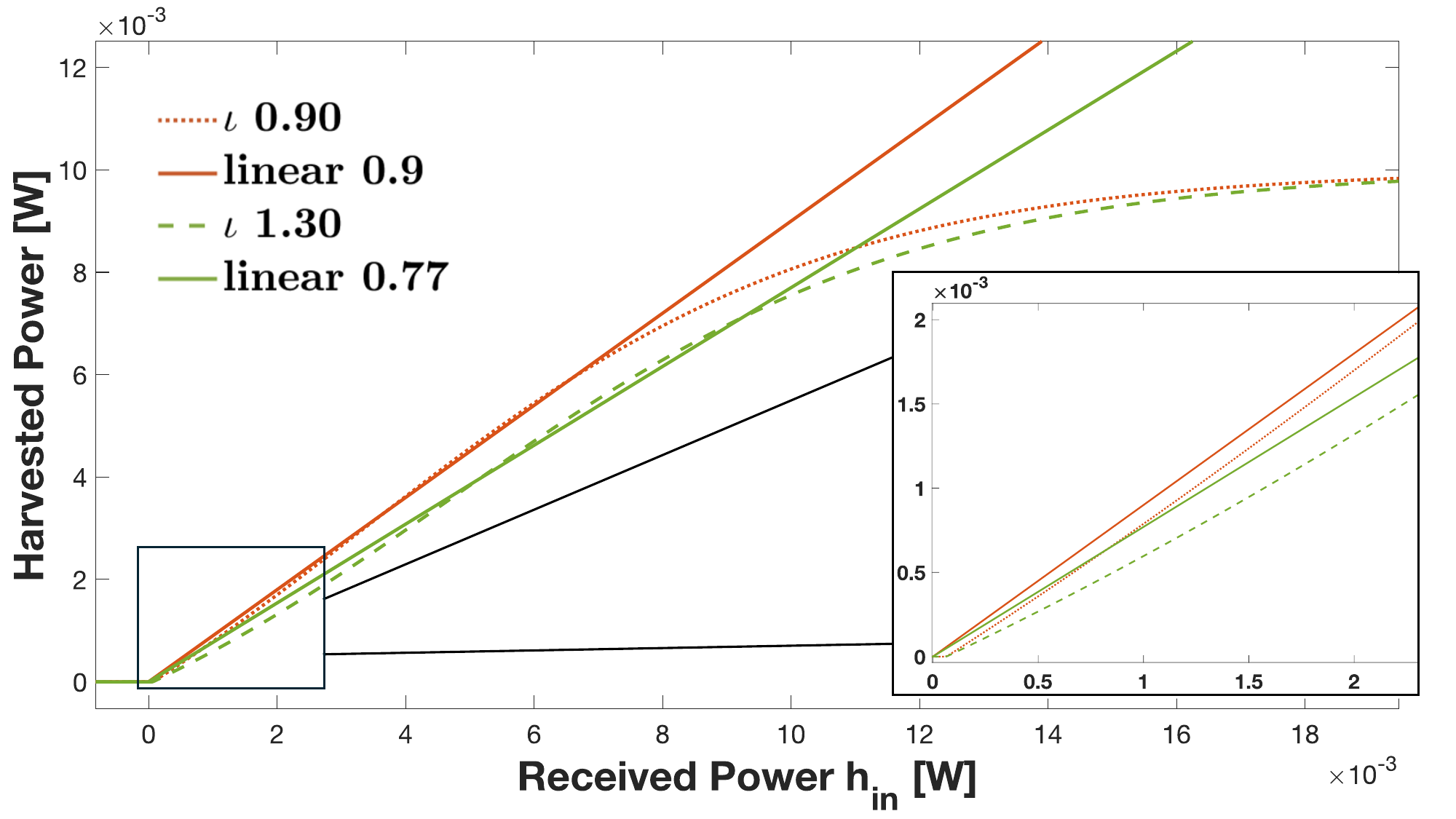}
\caption{Harvested power as a function of power received by IoT EH user, for a linear (with conversion efficiency of $0.9$ and $0.77$ respectively) as well as for the nonlinear model (Equation \ref{eq:sigmoid}) with $h_s=0.064$ mW, and $\chi=274$.}
\label{fig:mapping_function_nonlinear}
\end{minipage}%
\hspace{0.2in}
\begin{minipage}{.45\textwidth}
  \centering
  \includegraphics[width=\columnwidth]{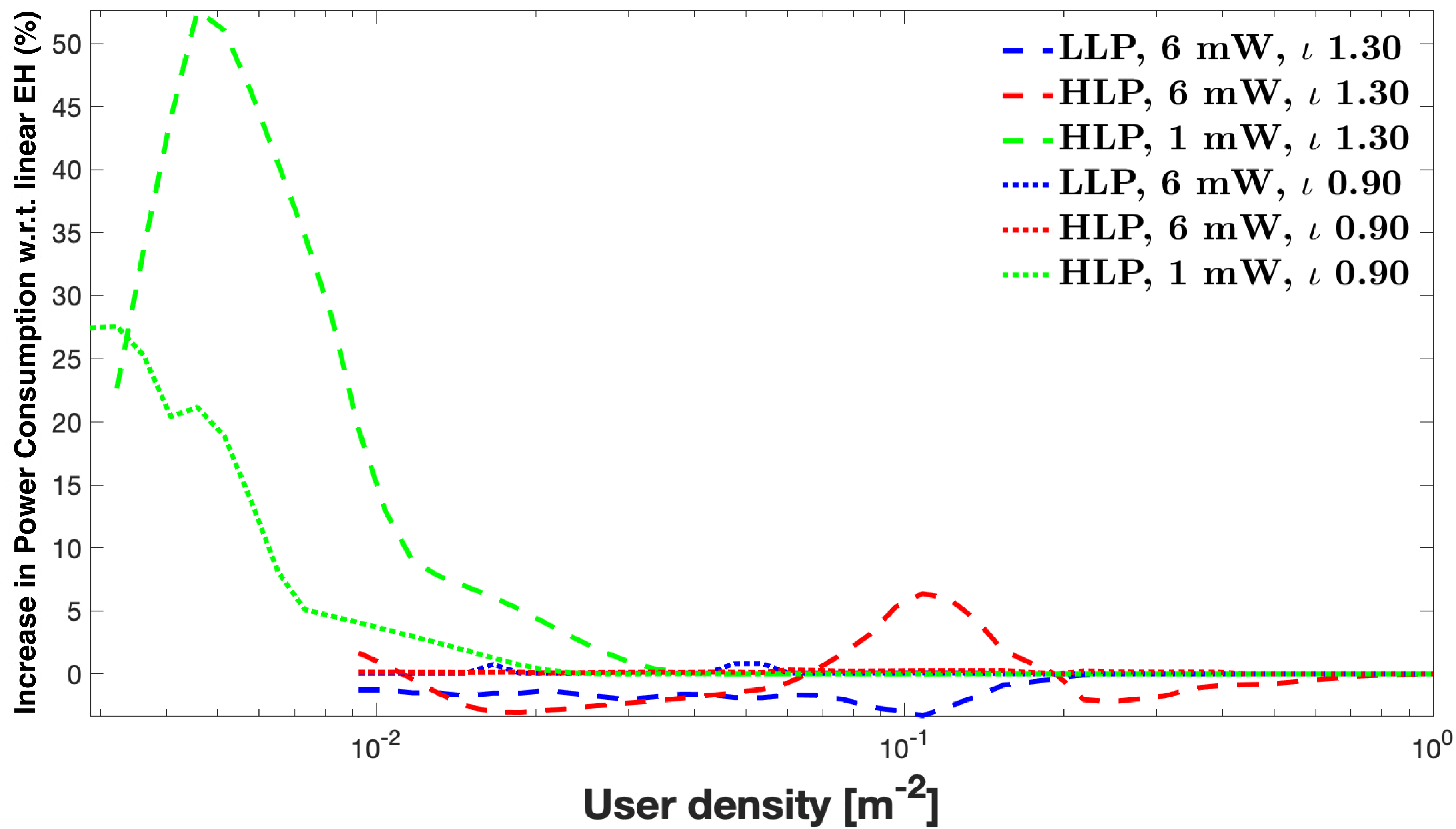}
\caption{Increase in consumed power with respect to linear EH, as a function of user density, for different values of the parameter $\iota$ of the nonlinear EH curve, and for different network setups, in the time splitting configuration.}
\label{fig:nonlinear_comparison}
\end{minipage}
\end{figure}

Figure \ref{fig:nonlinear_comparison} illustrates the increase in power consumed at the optimum induced by the nonlinear EH function versus that of its linear approximation, for the given values of $\iota$, and for time splitting EH mode.
As the plot suggests, nonlinearity brings to an increase of the consumed power for target minimum harvested power values which are much lower than $h_{max}$, and for low values of user density. For a target harvested power of $1$ mW the sigmoid is in a regime in which the conversion efficiency is $20\%$ to $30\%$ less than its maximum value, pushing the system to compensate with a more conservative transmit power allocation, thus bringing to a higher overall consumed power. This is more evident for lower user densities, as in those regimes the lower mean amount of users per base station brings (by the law of small numbers) to larger variance and thus to greatly amplify the effect of the $20\%$ to $30\%$ decrease in conversion efficiency on those users whose harvested power falls on the tail of the distribution. 

%% file: 08.Conclusions.tex
\section{Conclusions}
\label{sec:concl}
In this paper, we proposed an analytical framework for the characterization of the performance of a SWIPT network serving a combination of broadband users and EH IoT devices, and for the identification of energy-optimal network configurations satisfying constraints on user-perceived QoS. 
{A novel model that includes the usage of non-linear harvesting has been introduced taking into consideration aspects related to sensitivity and maximum capacity.}
Numerical results suggest that substantial energy savings are possible with schemes that adapt the main network parameters to fluctuations in user density. We also showed that such schemes remain essential, even in more energy-proportional cloud-RAN settings. The most interesting insight provided by our results concerns the interplay between energy harvesting and interference in cases of extreme densities of IoT devices. In such cases, we observed that very high device densities facilitate energy harvesting before interference becomes a problem. As a result, the optimal base station density can be very low even when the network load is high. This implies that base station sleep modes in SWIPT networks should follow different dynamics than those of networks that do not provide energy harvesting to IoT devices.

%% file: Appendix.tex
\appendix

\subsection{Proof of \lref{lemma:perjouledelay}}
\label{app:proof_lemma_perjouledelay}
\biagio{We derive the expression of $h_{in}$ only for the TS case, as the SPS and DPS cases are a minor variation of the former. The power the user receives at $x$ by its serving BS is $P D(x)^{-\alpha}$.
$K^{-1}(x)=N_{iot}(S(x))+w_d(1-\eta) N_{bb}(S(x))$ is the sum of the GPS weights of all users served by the base station serving the user at $x$. Given that $\eta$ is the GPS weight for the fraction of BS time dedicated to \textit{active} power transfer (to distinguish it from \textit{passive} power transfer, i.e. from power harvested by ambient radiation) to an IoT user, and as the GPS weight of IoT users is $1$, $\eta U_d(S(x))K(x)$ is the fraction of total base station time dedicated to actively charging the user at $x$. Thus, the amount of active transfer energy received by the user at $x$ is $P D(x)^{-\alpha} \eta U_d(S(x))G K(x)$.
This implies that the expression for the fraction of total base station time dedicated to transmitting information to the IoT user at $x$ is $(1-\eta) U_d(S(x))K(x)$, and that the BS time dedicated to serving other users is $(1-K(x))U_d(S(x))$. During this BS time, the BS charges also passively the given user, though it transmits to it with gain $L_g$.\\ 
From the above, it follows also that $\left[1-(1-\eta) U_d(S(x))K(x)\right]$ is the fraction of BS time during which the user at $x$ can harvest energy (from both active and passive transfers) because this is the fraction of time during which it is not receiving information, and $\left[1-(1-\eta) U_d(S(x))K(x)\right](I(x)+O(x))$ the received power due to passive power transfer from BSs other than the serving one, and from any transmitting user. Putting all together we have Eq. \ref{eq:harvested_power_TS}. }
\subsection{Proof of \thref{th:BE_tau}}
\label{app:proof_th_BE_tau}
\biagio{We consider the user at zero, but we drop this indication in what follows for ease of notation (i.e. $S(0)$ becomes $S$, and $D(0)$ becomes $D$).
 To derive $\bar{\tau}_d$, we start by computing the Palm expectation of $\tau_d^{id}(x)$ (we drop the dependence on P and G for notation ease):}
\[
\bar{\tau}_d=E^{0}\bigg[\frac{N_{iot}(S)+w_d N_{bb}(S)}{w_dC(D,I)}\bigg]= \int_0^{\infty}  E^0\bigg[\frac{N_{iot}(S,D)+w_d N_{bb}(S,D)}{w_dC(D,I)} |r \leq D \leq r + d r \bigg]P(r \leq D \leq r + d r)
\]
\[\approx\int_0^{\infty}  \frac{E^0[N_{iot}(S,D)+w_d N_{bb}(S,D)| r \leq D \leq r + d r ]}{w_dC(r,\bar{I}(r,k))}P(B(0,r) = \phi)\lambda_b 2 \pi r d r\]

\biagio{where $\bar{I}(r,k)$ is the average interfering power for the typical user at $r$, given by $\bar{I}(r,k)=\frac{PL\lambda_b2\pi r^{2-\alpha}}{k(\alpha-2)}\frac{\bar{\tau}_d}{\tau_d^0}$ \cite{rizzo2018esavingpotential}. The notation $N_{iot}(S,D)$, $N_{bb}(S,D)$ puts in evidence the fact that we are considering a BS located at $S$, serving the user at zero at a distance $D$.  
$P(B(0,r) = \phi)$ is the probability that a ball of radius $r$ centered at the origin is empty, equal to $e^{- \lambda_b \pi r^2}$.
If $N_{tot}(S,D)=N_{bb}(S,D)+N_{iot}(S,D)$, then $N_{iot}(S,D)+w_d N_{bb}(S,D)=QN_{tot}(S,D)$, with $Q=w_d+\gamma(\phi-w_d)$. The palm expectation at the numerator becomes $
QE^0[N_{tot}(S,D)| r \leq D \leq r + d r ]$.
The random variable $N_{tot}(S,D)$ is the total number of users present in the same cell as the user at the origin, when his distance from its serving base station is $D$.
As users are distributed according to a PPP with intensity $\lambda_u$, $N_{tot}(S,D)$ is Poisson, with an intensity given by 
the conditional Palm expectation inside the integral.}
Using Campbell's formula \cite{Stoyan1987}, this expectation becomes 
\[
E^0\bigg[\int_0^{\infty} \int_0^{2 \pi} \mathbf{1}_{(S(x,\theta) = S| r \leq D \leq r + d r)}\lambda_u d \theta x d x Q\bigg]=
\int_0^{\infty} \int_0^{2 \pi} \lambda_u P (S(x,\theta) = S| r \leq D \leq r + d r) d \theta x d x Q
\]
The conditional probability within the integral is given by
$\int_0^{\infty} \int_0^{2 \pi} \lambda_u e^{- \lambda_b A(r, x, \theta)} d \theta x d x$, where $A(r, x, \theta)$ is the area of the circle centered at $(x, \theta)$ with radius $x$ that is not overlapped by the circle centered at $(0, -r)$ with radius $r$
~\cite{rengarajan2015energy}.
By substituting, we get the expression for $\bar{\tau}_d$. The derivation of $\bar{\tau}_u$ follows along the same lines.

\biagio{Expressions (\ref{eq:tau_1}), (\ref{eq:tau_2}) and (\ref{eq:tau_3}) are implicit in the per bit delays. That is,  they are a function of the per-bit delays themselves, via the expression for the interference. Thus each constitutes a fixed point problem. However it is easy to see that the mapping operator associated with each of these fixed points is contractive, and thus the problem admits a unique solution.
Specifically, to prove the existence and uniqueness of the fixed point for Eq. \ref{eq:tau_1}, we prove that the operator $T(\bar{\tau}_d)$, whose expression is given by the right member of \eref{eq:tau_1}, is a contraction. To this end, we verify that Blackwell's sufficient conditions for a contraction \cite{blackwell1965discounted} hold for $T$. The monotonicity of $T$ is straightforward, because increasing $\bar{\tau}_d$ increases the mean BS utilization, and hence the mean interference level. This translates into an increase in per-bit delays (i.e. a decrease in mean throughput). For the discounting property, we prove that $\exists \beta\in(0,1)$ such that $T(\bar{\tau}_d+ a)\leq T(\bar{\tau}_d)+\beta a$, $\forall a \geq 0$ and for all system parameter values for which $\bar{\tau}_d$ is defined. This is done by upper bounding $\bar{I}(\bar{\tau}_d+ a)$ as the sum of two terms, one function of $\bar{\tau}_d$ and the other of $a$.  Substituting into \eref{eq:tau_1}, we have that the integrand in \eref{eq:tau_1} can be upper bounded as the sum of two terms, and after some algebraic steps this allows proving the discounting property. }

\subsection{Proof of \thref{th:BE_tau_stdev}}
\label{app:proof_stdev}

\begin{lemma}\label{lemma_recharging} The average power received from users for the typical user arriving in the system at a distance $r$ from the serving BS is approximated by $
\bar{O}=\frac{(1-\gamma)\delta_u P_{bb} +\phi \gamma P_I}{(1-\gamma)\delta_u+\phi\gamma}\frac{\lambda_b\pi \alpha }{\alpha-2}\frac{\bar{\tau}_u}{\tau_u^0}$.
\end{lemma}
\begin{proof}  
The transmit power from a user (averaging between IoT and BB users) is 
$\frac{(1-\gamma)\delta_u P_{bb} +\gamma\phi  P_I}{(1-\gamma)\delta_u+\gamma\phi}$.
Let $\chi(j)$ is the set of users served by the $j-$th BS, and $x'_i(t), i\in\chi(j)$ be the position of the $i$-th user served by BS $j$ at time t. Let $d(x,x'_i(t))$ denote the distance between the two users considered, and $u(x'_i(t))$ the probability that the given user is transmitting at time $t$. Then  
\vspace{-0.2in}
\be\label{eq:palminterf}
O(x,t)=\sum_{j}\sum_{i\in\chi(j)}\frac{(1-\gamma)\delta_u P_{bb} +\gamma \phi P_I}{(1-\gamma)\delta_u+\gamma \phi}d(x,x'_i(t)) ^{-\alpha}u(x'_i(t))
\ee
As we assume the base station has utilization $\frac{\bar{\tau}_u}{\tau_u^0}$ in the uplink,
$u(x'_i(t))=\frac{1}{\lambda_u A_j}\frac{\bar{\tau}_u}{\tau_u^0}$, where $A_j$ is the area of the Voronoi cell of BS j. Given that users are uniformly distributed in the plane, the Palm expectation of \eref{eq:palminterf} is well approximated by $\frac{(1-\gamma)\delta_u P_{bb} +\gamma \phi P_I}{(1-\gamma)\delta_u+\gamma \phi}E\left[\sum_{j}\frac{\int_{x'\in A_j} d(x,x')^{-\alpha} dx'}{\lambda_u A_j}\right]\frac{\bar{\tau}_u}{\tau_u^0}$. As $A_j$ is statistically independent on $d(x,x')$, the approximated formula becomes
\vspace{-0.3in}
\be\label{eq:quasifinal}
 \bar{O}=\frac{(1-\gamma)\delta_u P_{bb} +\gamma \phi P_I}{(1-\gamma)\delta_u+\gamma\phi}\frac{\int_0^{+\infty}\lambda_u min(1,s^{-\alpha})2\pi s ds}{\lambda_u \bar{A}}\frac{\bar{\tau}_u}{\tau_u^0}
\ee
where $\bar{A}=\lambda_b^{-1}$ is the mean area of a Voronoi cell for a BS density of $\lambda_b$.
The final expression of $\bar{O}$ is derived from the high attenuation assumption ($\alpha>2$).  
\end{proof}

\begin{proof}(\thref{th:BE_tau_stdev}) 
Let us rewrite expressions (5) to (7) in the form  $h(x)=F(x)+K(x)Z(x)$. We have:
\begin{itemize}
    \item In the \textit{TS} mode, $F_{TS}(x)=P D(x)^{-\alpha}LU_d(x) +I(x) +O(x)$ and $Z_{TS}(x)=P D(x)^{-\alpha}U_d(x)(G\eta-L_g)-U_d(x)(1-\eta)(I(x) +O(x))$.
\item In the \textit{SPS} mode, $F_{SPS}(x)=\nu F_{TS}(x)$ and $Z_{SPS}(x)=\nu P D(x)^{-\alpha}U_d(x)(G-L_g)$.
\item In the \textit{DPS} mode, $F_{DPS}(x)= F_{TS}(x)$ and $Z_{SPS}(x)=P D(x)^{-\alpha}U_d(x)(\nu G-L_g)$.
\end{itemize}
Let us consider first the TS case. 
Let's compute $E^0[F(D)| r \leq D \leq r + d r]$.
In the TS case, $\forall x$, we approximate $U_d(x)=\frac{\bar{\tau}_d}{\tau_d^0}$. Thus $
P E^0[D(x)^{-\alpha}LU_d(x)| r \leq D \leq r + d r]=P r^{-\alpha}L_g\frac{\bar{\tau}_d}{\tau_d^0}$.
As for the remaining terms, 
$O(x)$ does not depend on the user distance from its serving BS. The expression of its expected value is thus given by \lref{lemma_recharging}. As for the conditional expectation of $I(x)$, it is given by \eref{eq:interference} in \thref{th:BE_tau}, multiplied by the reuse factor $k$. We thus have $E^0[I(x) + O(x)| r \leq D \leq r + d r]=
k\bar{I}(r,k) +\bar{O}$.
The SPS and DPS cases are derived similarly.
As for $K(x)Z_{TS}(x)$,
 $Z(x)$ is a function of $D$, except for $O(x)$, for which we approximate  $E^0\left[K(D)O(D)| r \leq D \leq r + d r \right]= E^0\left[K(D)| r \leq D \leq r + d r \right]O(r)$. We thus have
\[
E^0\left[ K(D)| r \leq D \leq r + d r \right] Z(r)
=E^0\left[\frac{1}{Q N_{tot}(D)}| r \leq D \leq r + d r \right]Z(r)
\]
For Jensen's inequality, \scalebox{1}[1]{$E^0\left[\frac{1}{N_{tot}(D)}| r \leq D \leq r + d r \right]\leq \frac{1}{E^0[N_{tot}(D)| r \leq D \leq r + d r]}
$}.
For a given cell, $N_{tot}(D)$ is Poisson distributed. Thus the ratio of the standard deviation over the mean of this variable decreases with increasing user density. Therefore in the dense IoT regime, such an inequality is tight. therefore, the denominator is then computed as a function of $r$ as in the proof of \thref{th:BE_tau}. 
The derivation of the expressions for the DPS and SPS cases follows along the same line.
\end{proof}

%% file: 00.main.bbl
\begin{thebibliography}{10}
\providecommand{\url}[1]{#1}
\csname url@samestyle\endcsname
\providecommand{\newblock}{\relax}
\providecommand{\bibinfo}[2]{#2}
\providecommand{\BIBentrySTDinterwordspacing}{\spaceskip=0pt\relax}
\providecommand{\BIBentryALTinterwordstretchfactor}{4}
\providecommand{\BIBentryALTinterwordspacing}{\spaceskip=\fontdimen2\font plus
\BIBentryALTinterwordstretchfactor\fontdimen3\font minus
  \fontdimen4\font\relax}
\providecommand{\BIBforeignlanguage}[2]{{%
\expandafter\ifx\csname l@#1\endcsname\relax
\typeout{** WARNING: IEEEtran.bst: No hyphenation pattern has been}%
\typeout{** loaded for the language `#1'. Using the pattern for}%
\typeout{** the default language instead.}%
\else
\language=\csname l@#1\endcsname
\fi
#2}}
\providecommand{\BIBdecl}{\relax}
\BIBdecl

\bibitem{9930622}
G.~Rizzo, M.~{Ajmone Marsan}, and C.~Esposito, ``Energy-optimal ran
  configurations for swipt iot,'' in \emph{2022 20th International Symposium on
  Modeling and Optimization in Mobile, Ad hoc, and Wireless Networks (WiOpt)},
  2022, pp. 169--176.

\bibitem{Rajab2023}
H.~Rajab and T.~Cinkler, \emph{Enhanced Energy Efficiency and Scalability in
  Cellular Networks for Massive IoT}.\hskip 1em plus 0.5em minus 0.4em\relax
  Singapore: Springer Nature Singapore, 2023, pp. 283--305.

\bibitem{PrepareY71:online}
``Prepare yourself for the “tsunami of data” expected to hit by 2025,''
  \url{https://futurism.com/prepare-yourself-tsunami-data-expected-hit-2025},
  (Accessed on 05/09/2022).

\bibitem{TAHAEI2020102538}
H.~Tahaei, F.~Afifi, A.~Asemi, F.~Zaki, and N.~B. Anuar, ``The rise of traffic
  classification in iot networks: A survey,'' \emph{Journal of Network and
  Computer Applications}, vol. 154, p. 102538, 2020.

\bibitem{wan9206046}
Y.~Wan, K.~Xu, F.~Wang, and G.~Xue, ``Characterizing and mining traffic
  patterns of iot devices in edge networks,'' \emph{IEEE Transactions on
  Network Science and Engineering}, vol.~8, no.~1, pp. 89--101, 2021.

\bibitem{sanislav2021energy}
T.~Sanislav, G.~D. Mois, S.~Zeadally, and S.~C. Folea, ``Energy harvesting
  techniques for internet of things (iot),'' \emph{IEEE Access}, vol.~9, pp.
  39\,530--39\,549, 2021.

\bibitem{varshney2008transporting}
L.~R. Varshney, ``Transporting information and energy simultaneously,'' in
  \emph{2008 IEEE international symposium on information theory}.\hskip 1em
  plus 0.5em minus 0.4em\relax IEEE, 2008, pp. 1612--1616.

\bibitem{SWIPT_costanzo2021evolution}
A.~Costanzo, D.~Masotti, G.~Paolini, and D.~Schreurs, ``Evolution of swipt for
  the iot world: Near-and far-field solutions for simultaneous wireless
  information and power transfer,'' \emph{IEEE Microwave Magazine}, vol.~22,
  no.~12, pp. 48--59, 2021.

\bibitem{SWIPT_huang2017simultaneous}
J.~Huang, C.-C. Xing, and C.~Wang, ``Simultaneous wireless information and
  power transfer: Technologies, applications, and research challenges,''
  \emph{IEEE Communications Magazine}, vol.~55, no.~11, pp. 26--32, 2017.

\bibitem{SWIPT_ozyurt2022survey}
S.~{\"O}zyurt, A.~Co{\c{s}}kun, S.~B{\"u}y{\"u}k{\c{c}}orak, G.~K. Kurt, and
  O.~Kucur, ``A survey on multiuser swipt communications for 5g+,'' \emph{IEEE
  Access}, vol.~10, pp. 109\,814--109\,849, 2022.

\bibitem{SWIPT_perera2017simultaneous}
T.~D.~P. Perera, D.~N.~K. Jayakody, S.~K. Sharma, S.~Chatzinotas, and J.~Li,
  ``Simultaneous wireless information and power transfer (swipt): Recent
  advances and future challenges,'' \emph{IEEE Communications Surveys \&
  Tutorials}, vol.~20, no.~1, pp. 264--302, 2017.

\bibitem{zhou_wireless_2013}
X.~Zhou, R.~Zhang, and C.~K. Ho, ``\BIBforeignlanguage{en}{Wireless
  {Information} and {Power} {Transfer}: {Architecture} {Design} and
  {Rate}-{Energy} {Tradeoff}},'' \emph{\BIBforeignlanguage{en}{IEEE Trans.
  Commun.}}, vol.~61, no.~11, pp. 4754--4767, Nov. 2013.

\bibitem{huang_energy-efficient_2018}
Y.~Huang, M.~Liu, and Y.~Liu, ``\BIBforeignlanguage{en}{Energy-{Efficient}
  {SWIPT} in {IoT} {Distributed} {Antenna} {Systems}},''
  \emph{\BIBforeignlanguage{en}{IEEE Internet Things J.}}, vol.~5, no.~4, pp.
  2646--2656, Aug. 2018.

\bibitem{Tran8388301}
H.-V. Tran, G.~Kaddoum, and K.~T. Truong, ``Resource allocation in swipt
  networks under a nonlinear energy harvesting model: Power efficiency, user
  fairness, and channel nonreciprocity,'' \emph{IEEE Transactions on Vehicular
  Technology}, vol.~67, no.~9, pp. 8466--8480, 2018.

\bibitem{lu_energy_2021}
W.~Lu, X.~Xu, G.~Huang, B.~Li, Y.~Wu, N.~Zhao, and F.~R. Yu,
  ``\BIBforeignlanguage{en}{Energy {Efficiency} {Optimization} in {SWIPT}
  {Enabled} {WSNs} for {Smart} {Agriculture}},''
  \emph{\BIBforeignlanguage{en}{IEEE Trans. Ind. Inf.}}, vol.~17, no.~6, pp.
  4335--4344, Jun. 2021.

\bibitem{Lee9031335}
K.~Lee and W.~Lee, ``Learning-based resource management for swipt,'' \emph{IEEE
  Systems Journal}, vol.~14, no.~4, pp. 4750--4753, 2020.

\bibitem{lam2016system}
T.~T. Lam, M.~Di~Renzo, and J.~P. Coon, ``System-level analysis of swipt mimo
  cellular networks,'' \emph{IEEE Communications Letters}, vol.~20, no.~10, pp.
  2011--2014, 2016.

\bibitem{di_renzo_system-level_2017}
M.~Di~Renzo, ``\BIBforeignlanguage{en}{System-{Level} {Analysis} and
  {Optimization} of {Cellular} {Networks} {With} {Simultaneous} {Wireless}
  {Information} and {Power} {Transfer}: {Stochastic} {Geometry} {Modeling}},''
  \emph{\BIBforeignlanguage{en}{IEEE Trans. Veh. Technol.}}, vol.~66, no.~3,
  pp. 2251--2275, Mar. 2017.

\bibitem{Baccelli_stochasticgeometry}
F.~Baccelli, B.~B{\l}aszczyszyn \emph{et~al.}, ``{Stochastic geometry and
  wireless networks: Volume II Applications},'' \emph{Foundations and Trends in
  Networking}, vol.~4, no. 1--2, pp. 1--312, 2010.

\bibitem{niyato2017wireless}
D.~Niyato, D.~I. Kim, M.~Maso, and Z.~Han, ``Wireless powered communication
  networks: Research directions and technological approaches,'' \emph{IEEE
  Wireless Communications}, vol.~24, no.~6, pp. 88--97, 2017.

\bibitem{Clerckx8476597}
B.~Clerckx, R.~Zhang, R.~Schober, D.~W.~K. Ng, D.~I. Kim, and H.~V. Poor,
  ``Fundamentals of wireless information and power transfer: From rf energy
  harvester models to signal and system designs,'' \emph{IEEE Journal on
  Selected Areas in Communications}, vol.~37, no.~1, pp. 4--33, 2019.

\bibitem{luo_study_2021}
Y.~Luo, C.~Luo, G.~Min, G.~Parr, and S.~McClean, ``\BIBforeignlanguage{en}{On
  the {Study} of {Sustainability} and {Outage} of {SWIPT}-{Enabled} {Wireless}
  {Communications}},'' \emph{\BIBforeignlanguage{en}{IEEE J. Sel. Top. Signal
  Process.}}, vol.~15, no.~5, pp. 1159--1168, Aug. 2021.

\bibitem{havutran}
H.-V. Tran, G.~Kaddoum, and K.~T. Truong, ``Resource allocation in swipt
  networks under a nonlinear energy harvesting model: Power efficiency, user
  fairness, and channel nonreciprocity,'' \emph{IEEE Transactions on Vehicular
  Technology}, vol.~67, no.~9, pp. 8466--8480, 2018.

\bibitem{tang_energy_2018-1}
J.~Tang, D.~K.~C. So, N.~Zhao, A.~Shojaeifard, and K.-K. Wong,
  ``\BIBforeignlanguage{en}{Energy {Efficiency} {Optimization} {With} {SWIPT}
  in {MIMO} {Broadcast} {Channels} for {Internet} of {Things}},''
  \emph{\BIBforeignlanguage{en}{IEEE Internet Things J.}}, vol.~5, no.~4, pp.
  2605--2619, Aug. 2018.

\bibitem{tang_energy_2019}
J.~Tang, J.~Luo, M.~Liu, D.~K.~C. So, E.~Alsusa, G.~Chen, K.-K. Wong, and J.~A.
  Chambers, ``\BIBforeignlanguage{en}{Energy {Efficiency} {Optimization} for
  {NOMA} {With} {SWIPT}},'' \emph{\BIBforeignlanguage{en}{IEEE J. Sel. Top.
  Signal Process.}}, vol.~13, no.~3, pp. 452--466, Jun. 2019.

\bibitem{yuan_energy_2019}
Y.~Yuan, Y.~Xu, Z.~Yang, P.~Xu, and Z.~Ding, ``\BIBforeignlanguage{en}{Energy
  {Efficiency} {Optimization} in {Full}-{Duplex} {User}-{Aided} {Cooperative}
  {SWIPT} {NOMA} {Systems}},'' \emph{\BIBforeignlanguage{en}{IEEE Trans.
  Commun.}}, vol.~67, no.~8, pp. 5753--5767, Aug. 2019.

\bibitem{li_robust_2020}
Q.~Li and L.~Yang, ``\BIBforeignlanguage{en}{Robust {Optimization} for {Energy}
  {Efficiency} in {MIMO} {Two}-{Way} {Relay} {Networks} {With} {SWIPT}},''
  \emph{\BIBforeignlanguage{en}{IEEE Systems Journal}}, vol.~14, no.~1, pp.
  196--207, Mar. 2020.

\bibitem{Chu9463400}
M.~Chu, A.~Liu, J.~Chen, V.~K.~N. Lau, and S.~Cui, ``A stochastic geometry
  analysis for energy-harvesting-based device-to-device communication,''
  \emph{IEEE Internet of Things Journal}, vol.~9, no.~2, pp. 1591--1607, 2022.

\bibitem{nasir2013relaying}
A.~A. Nasir, X.~Zhou, S.~Durrani, and R.~A. Kennedy, ``Relaying protocols for
  wireless energy harvesting and information processing,'' \emph{IEEE
  Transactions on Wireless Communications}, vol.~12, no.~7, pp. 3622--3636,
  2013.

\bibitem{wang_wirelessly-powertwoway_2017}
S.~Wang, M.~Xia, K.~Huang, and Y.-C. Wu, ``Wirelessly powered two-way
  communication with nonlinear energy harvesting model: Rate regions under
  fixed and mobile relay,'' \emph{IEEE Transactions on Wireless
  Communications}, vol.~16, no.~12, pp. 8190--8204, 2017.

\bibitem{futureproof2015}
B.~Debaillie, C.~Desset, and F.~Louagie, ``A flexible and future-proof power
  model for cellular base stations,'' in \emph{Vehicular Technology Conference
  (VTC Spring), 81st}.\hskip 1em plus 0.5em minus 0.4em\relax IEEE, 2015, pp.
  1--7.

\bibitem{mushtaq_power_2017}
M.~S. Mushtaq, S.~Fowler, and A.~Mellouk, ``\BIBforeignlanguage{en}{Power
  saving model for mobile device and virtual base station in the {5G} era},''
  in \emph{\BIBforeignlanguage{en}{{IEEE} {ICC}}}, Paris, France, May 2017, pp.
  1--6.

\bibitem{Stoyan1987}
D.~Stoyan, W.~S. Kendall, and J.~Mecke, \emph{Stochastic geometry and its
  applications}.\hskip 1em plus 0.5em minus 0.4em\relax Wiley, 1987.

\bibitem{BANAFAA2023245}
M.~Banafaa, I.~Shayea, J.~Din, M.~{Hadri Azmi}, A.~Alashbi, Y.~{Ibrahim
  Daradkeh}, and A.~Alhammadi, ``6g mobile communication technology:
  Requirements, targets, applications, challenges, advantages, and
  opportunities,'' \emph{Alexandria Engineering Journal}, vol.~64, pp.
  245--274, 2023.

\bibitem{QADIR2023296}
Z.~Qadir, K.~N. Le, N.~Saeed, and H.~S. Munawar, ``Towards 6g internet of
  things: Recent advances, use cases, and open challenges,'' \emph{ICT
  Express}, vol.~9, no.~3, pp. 296--312, 2023.

\bibitem{9369324}
F.~Guo, F.~R. Yu, H.~Zhang, X.~Li, H.~Ji, and V.~C.~M. Leung, ``Enabling
  massive iot toward 6g: A comprehensive survey,'' \emph{IEEE Internet of
  Things Journal}, vol.~8, no.~15, pp. 11\,891--11\,915, 2021.

\bibitem{9509294}
D.~C. Nguyen, M.~Ding, P.~N. Pathirana, A.~Seneviratne, J.~Li, D.~Niyato,
  O.~Dobre, and H.~V. Poor, ``6g internet of things: A comprehensive survey,''
  \emph{IEEE Internet of Things Journal}, vol.~9, no.~1, pp. 359--383, 2022.

\bibitem{rajwar2023exhaustive}
K.~Rajwar, K.~Deep, and S.~Das, ``An exhaustive review of the metaheuristic
  algorithms for search and optimization: taxonomy, applications, and open
  challenges,'' \emph{Artificial Intelligence Review}, pp. 1--71, 2023.

\bibitem{mitchell1998introduction}
M.~Mitchell, \emph{An introduction to genetic algorithms}.\hskip 1em plus 0.5em
  minus 0.4em\relax MIT press, 1998.

\bibitem{chipperfield1995matlab}
A.~Chipperfield and P.~Fleming, ``The matlab genetic algorithm toolbox,'' 1995.

\bibitem{489178}
R.~Hinterding, ``Gaussian mutation and self-adaption for numeric genetic
  algorithms,'' in \emph{Proceedings of 1995 IEEE International Conference on
  Evolutionary Computation}, vol.~1, 1995, pp. 384--.

\bibitem{bazaraa2013nonlinear}
M.~S. Bazaraa, H.~D. Sherali, and C.~M. Shetty, \emph{Nonlinear programming:
  theory and algorithms}.\hskip 1em plus 0.5em minus 0.4em\relax John wiley \&
  sons, 2013.

\bibitem{lin2021data}
J.~Lin, Y.~Chen, H.~Zheng, M.~Ding, P.~Cheng, and L.~Hanzo, ``A data-driven
  base station sleeping strategy based on traffic prediction,'' \emph{IEEE
  Transactions on Network Science and Engineering}, 2021.

\bibitem{ge2017energy}
X.~Ge, J.~Yang, H.~Gharavi, and Y.~Sun, ``Energy efficiency challenges of 5g
  small cell networks,'' \emph{IEEE communications Magazine}, vol.~55, no.~5,
  pp. 184--191, 2017.

\bibitem{rizzo2018esavingpotential}
G.~A. Rizzo and M.~{Ajmone Marsan}, ``The energy saving potential of static and
  adaptive resource provisioning in dense cellular networks,'' in \emph{{IEEE
  COMSNETS}}, 2018, pp. 297--304.

\bibitem{rengarajan2015energy}
B.~Rengarajan, G.~Rizzo, and M.~{Ajmone Marsan}, ``Energy-optimal base station
  density in cellular access networks with sleep modes,'' \emph{Computer
  Networks}, vol.~78, pp. 152--163, 2015.

\bibitem{blackwell1965discounted}
D.~Blackwell, ``Discounted dynamic programming,'' \emph{The Annals of
  Mathematical Statistics}, vol.~36, no.~1, pp. 226--235, 1965.

\end{thebibliography}
